\documentclass{IEEEtaes}
\usepackage{amsmath,amsfonts}
\usepackage{algorithmic}
\usepackage{algorithm}
\usepackage{lineno,hyperref}
\usepackage{multirow}
\usepackage{epstopdf}
\usepackage{graphicx}
\usepackage{float}
\usepackage{caption}
\usepackage{stfloats}
\usepackage{changes}
\usepackage{color,array,amsthm}
\usepackage{amssymb}
\usepackage{graphicx}
\usepackage{eucal}
\usepackage{mathrsfs}
\usepackage{hyperref}      
\jvol{XX}
\jnum{XX}
\jmonth{XXXXX}
\paper{1234567}
\pubyear{2020}
\doiinfo{TAES.2020.Doi Number}

\newtheorem{theorem}{Theorem}

\newtheorem{remark}{Remark}
\newtheorem{assumption}{Assumption}
\usepackage[inkscapelatex=false]{svg}
\usepackage[numbers,sort&compress]{natbib}
\usepackage{soul} 
\usepackage{color, xcolor} 

\soulregister{\cite}7 
\soulregister{\citep}7 
\soulregister{\citet}7 
\soulregister{\ref}7 
\soulregister{\pageref}7 
\modulolinenumbers[5]

\begin{document}

\title{Composite Triggered Intermittent Control for Constrained Spacecraft Attitude Tracking}

\author{Jiakun Lei}
\affil{Zhejiang University, China} 

\author{Tao Meng}
\affil{Zhejiang University, China\\ Hainan Research Institute of Zhejiang University, China} 

\author{Kun Wang}
\affil{Zhejiang University, China}

\author{Weijia Wang}
\affil{Zhejiang University, China}

\author{Shujian Sun}
\affil{Zhejiang University, China}


\authoraddress{Jiakun Lei, Kun Wang, Weijia Wang, Shujian Sun are with the School of Aeronautics and Astronautics, Zhejiang University, Hangzhou 310027, China (Email: leijiakun@zju.edu.cn; wang\_kun@zju.edu.cn; weijiawang@zju.edu.cn; sunshujian@zju.edu.cn). Tao Meng(Corresponding Author) is with the School of Aeronautics and Astronautics, Zhejiang University, Hangzhou 310027, China, and Hainan Reserach Institute of Zhejiang University, Sanya, 572025, China (Email: mengtao@zju.edu.cn).}

\markboth{Lei ET AL.}{SAPPC of Attitude Tracking}
\maketitle

\begin{abstract}
	This paper focuses on the spacecraft attitude control problem with intermittent actuator activation, taking into account the attitude rotation rate limitation and input saturation issue simultaneously. To address this problem, we first propose a composite event-trigger mechanism, which composed of two state-dependent trigger that governing the activation and deactivation of actuators. Subsequently, by introducing the cascaded decomposition of Backstepping control philosophy, the designed trigger mechanism is then applied to the decomposed dynamical subsystem, providing a layered intermittent stabilization strategy.
	Further, the basic intermittent attitude controller is extended to a "constrained version" by introducing a strictly bounded virtual control law and an input saturation compensation auxiliary system.
	By analyzing the local boundedness of the system on each inter-event time interval, a uniformly, strictly decreasing upper boundary of the lumped system is further characterized, thereby completing the proof of the system's uniformly ultimately boundedness (UUB). Finally, numerical simulation results are illustrated to demonstrate the effectiveness of the proposed scheme.
\end{abstract}

\begin{IEEEkeywords}			
	Constrained Attitude Control; Intermittent Control; Event-trigger control	
\end{IEEEkeywords}
%

\section{INTRODUCTION}
The Intermittent Control (IC) problem has become a topic of significant interest in recent years due to its practical meaning for real control systems \cite{zochowski2000intermittent,hu2013sliding}. This problem is particularly relevant for spacecraft attitude control systems, as actual actuators intermittently exert control efforts at discrete time intervals with a calculated control width, instead of continuously or impulsively.
On the other hand, owing to the fact that spacecraft spend most of their lifecycle in a steady-state maintenance mode, the conventional periodic control strategy may be unnecessary and overly conservative. Meanwhile, since control actuators own limited lifespan, conventional periodic control strategy may not be beneficial for longing its service life.
These problems highlight the need for a control policy to achieve the desired attitude control while considerably reducing the actuator's acting frequency.
Motivated by these problems, this paper aims to develop an intermittent control framework for spacecraft attitude tracking. To facilitate the proposed scheme's potential application, the attitude rotation rate limitations and input saturation issue is further considered simultaneously.

For existing researches on the intermittent control strategy, there are mainly two different technical paths, classified as periodically intermittent control (PIC) and aperiodically intermittent control (AIC).
The Periodically Intermittent Control (PIC) approach views the system as a time-dependent switching system that exhibits different closed-loop dynamics based on fixed control duration and non-control rest time. This perspective has been explored for stabilizing general nonlinear systems and network systems in previous studies (e.g., \cite{li2007stabilization, chen2016delay, li2007exponential, huang2009stabilization,yu2012synchronization,li2016complete}). Typically \cite{li2007stabilization}, the intermittent controller is modeled to be switched between 0 and a proportional feedback control law (e.g., $u=Ke(t)$), and the sufficient stability condition is further derived using Lyapunov-based analysis and Linear Matrix Inequality (LMI) methodology.
However, since the PIC approach strictly follows a pre-calculated trigger sequence without considering the current circumstance, it may have conservatism problem.

In order to remove the limitation problem of the PIC approach, the Aperiodically Intermittent Control (AIC) approach has been developed to introduce flexibility and adaptability to switching conditions \cite{liang2018exponential,liu2021exponential}. This approach has been investigated in several studies, including nonlinear dynamical systems, complex networks, and synchronization of coupling systems, as discussed in \cite{liang2018exponential,liu2021exponential,liu2015synchronization}.
The Time-dependent AIC control (T-AIC) for nonlinear dynamical systems is presented in \cite{liu2021exponential}, which utilizes the average dwell-time condition technique stated in \cite{hespanha2008lyapunov,dashkovskiy2017input,liu2018input} to derive a sufficient condition for exponential stability. Although T-AIC methods provide less conservatism than PIC methods, they still lack a fully direct consideration of current status, as mentioned by previous research \cite{liu2020stabilization}.

Meanwhile, the development of the event-triggered control has raised significant interest, as stated by numerous references \cite{girard2014dynamic,eqtami2010event,wu2018event,liu2020event}. In regarding of this, the state-dependent Event-triggered Aperiodically Intermittent Control (E-AIC) scheme was firstly presented in \cite{liu2021exponential}, which provides a threshold-dependent trigger mechanism to check if the exponential convergence condition is satisfied or not. Nevertheless, it should be noted that this approach is only partially event-triggered", as the control duration (i.e., control width) is a fixed parameter that selected in advance. Since the trigger condition that is suitable for the transition process may differ from that suitable for steady-state control, the presented trigger mechanism may lead to high-frequency triggering. This point of view is also mentioned in \cite{liu2023dynamic}. 

Recently, a novel idea has been proposed in \cite{ong2022stability}, which can be classified as an E-AIC method. This method utilizes a certificate bound to determine whether the actuator should be activated or deactivated \cite{koga2022event}. By imposing a predefined upper boundary on the time evolution of the Lyapunov certificate, this approach ensures that it remains within an admissible asymptotically-converged bound, thereby guaranteeing stability even when actuators are switched off. This E-AIC methodology has been further applied in \cite{ong2022stability} to develop a CLF-QP based controller.

Although the mentioned work that stated in \cite{ong2022stability} provides a novel perspective, there are still some problems that are worth noticing. Specifically, the control methodology presented in the literature guarantees stability through a CLF-QP (Control Lyapunov Function Quadratic Programming) based approach, and the satisfaction of constraints is achieved through a control barrier function technique. However, this may result in a heavy computational burden, which is a significant concern for micro-satellites with limited computing capability. Therefore, it is meaningful to investigate whether the main idea of such a fully event-triggered framework can be integrated with various general nonlinear control frameworks. Additionally, the feasibility of constrained attitude control using this approach warrants further investigation.

Motivated by these issues, this paper aims to develop a direct method to yield an intermittent control strategy compatible with general nonlinear control frameworks. 
The main contribution can be concluded as follows:

\textbf{1.} Firstly, this paper proposes a composite event-trigger mechanism, which efficiently handles the alternatively switching of actuator's turn on and off. Subsequently, by further combining this framework with backstepping control philosophy, we further provide a complete layered intermittent stabilizing strategy, which enable us to realize the attitude control with much lower actuator acting frequency. In this paper, we simply extended the controller to a "constrained attitude control" version, which further validates the extensibility of the proposed framework.

\textbf{2.} Motivated by the analysis method in switching control theory, a systematic approach for the stability evaluation of such a nonlinear intermittent controller is firstly presented.
Specifically, the proof of the local boundedness of the lumped Lyapunov trajectory is provided both for turn-on and turn-off inter-event time, suggesting that the system's time-evolution trajectory is strictly beneath an exponentially-convergedn function. By further utilizing a characteristic of exponentially-converged functions, a uniformly continuous boundary of the system's Lyapunov trajectory is further characterized over the entire time domain, thereby proving the system's UUB (Uniformly Ultimate Boundedness) property. Additionally, the system's residual set is strongly related to the design parameter of the trigger mechanism. This provides a practical suggestion for parameter selection, which enhances the potential application value of the proposed scheme.

\textit{Notations}
This paper defines the following notations for the upcoming analysis: 
$\|\cdot\|$ represents the induced norm of arbitrary given matrix, or the Euclidean norm of any given vector. The $i$ th component of arbitrary given vector $\boldsymbol{x}$ is represented as $x_{i}$. The element-spanned diagonal matrix is denoted as $\text{diag}\left(a_{i}\right)$, whose diagonal elements are sequentially given as $a_{1},...a_{i}$. $\text{vec}\left(a_{i}\right)$ represents the element-spanned column vector, corresponding to $\left[a_{1},...a_{i}\right]^{\text{T}}$. $\boldsymbol{a}^{\times}\in\mathbb{R}^{3\times 3}$ represents the corresponding cross manipulation matrix of $\boldsymbol{a}\in\mathbb{R}^{3}$, such that $\boldsymbol{a}^{\times}\boldsymbol{b} = \boldsymbol{a} \times \boldsymbol{b}$ holds ($\boldsymbol{b}\in\mathbb{R}^{3}$). Correspondingly, the $n\times n$ identity matrix is denoted as $\boldsymbol{I}_{n}\in\mathbb{R}^{n\times n}$. Further, we denote the Earth-Central-Inertial (ECI) Frame as $\mathfrak{R}_{i}$, while the spacecraft body-fixed frame is denoted by $\mathfrak{R}_{b}$.

\section{PROBLEM FORMULATION}\label{secproblem}
\subsection{System Modeling}
Considering the attitude error kinematics and dynamics equation modeled with the unit attitude quaternion, the attitude error system can be expressed as follows\cite{lei2023singularity}:
\begin{equation}\label{sys}
	\begin{aligned}
		\dot{\boldsymbol{q}}_{ev} &= \boldsymbol{\varGamma}_{e}\boldsymbol{\omega}_{e}\\
		\dot{q}_{e0} &= -\frac{1}{2}\boldsymbol{q}^{\text{T}}_{ev}\boldsymbol{\omega}_{e}\\ 
		\boldsymbol{J}\dot{\boldsymbol{\omega}}_{e} &= \boldsymbol{\Omega}_{e} + \boldsymbol{u} + \boldsymbol{d}
	\end{aligned}
\end{equation}

Let $\boldsymbol{q}_{e} \triangleq \left[\boldsymbol{q}^{\text{T}}_{ev},q_{e0}\right]^{\text{T}}\in\mathbb{R}^{4}$ denotes the attitude error quaternion of the spacecraft with respect to the inertial frame $\mathfrak{R}_{i}$, where $\boldsymbol{q}_{ev} = \left[q_{ev1},...q_{ev3}\right]^{\text{T}}\in\mathbb{R}^{3}$ and $q_{e0}\in\mathbb{R}$ stands for the vector part and the scalar part of $\boldsymbol{q}_{e}$, respectively. The error angular velocity of the spacecraft with respect to the inertial frame $\mathfrak{R}_{i}$ is denoted as $\boldsymbol{\omega}_{e} \triangleq \boldsymbol{\omega}_{s} - \boldsymbol{C}_{e}\boldsymbol{\omega}_{d} \in\mathbb{R}^{3}$, expressed in the current body-fixed frame $\mathfrak{R}_{b}$. Here $\boldsymbol{\omega}_{s}\in\mathbb{R}^{3}$ stands for the current angular velocity expressed in frame $\mathfrak{R}_{b}$, and $\boldsymbol{\omega}_{d}\in\mathbb{R}^{3}$ denotes the desired angular velocity expressed in the target attitude, $\boldsymbol{C}_{e}\in\mathbb{R}^{3\times 3}$ represents the transformation matrix from the target body-fixed frame to the current one. $\boldsymbol{u}\in\mathbb{R}^{3}$ stands for the actual exerted actuator output, and the system's lumped disturbance is denoted as $\boldsymbol{d}$. Additionally,  $\boldsymbol{J}\in\mathbb{R}^{3\times 3}$ stands for the inertia matrix, 
$\boldsymbol{\Omega}_{e}$ denotes a lumped dynamical term, defined as
$\boldsymbol{\Omega}_{e} \triangleq \boldsymbol{J}\boldsymbol{\omega}^{\times}_e\boldsymbol{C}_e\boldsymbol{\omega}_d 
- \boldsymbol{J}\boldsymbol{C}_e\dot{\boldsymbol{\omega}}_d
-\boldsymbol{\omega}_s^{\times}\boldsymbol{J}\boldsymbol{\omega}_s \in\mathbb{R}^{3}$. The Jacobian matrix of the error kinematics equation is defined as $\boldsymbol{\varGamma}_{e} = \frac{1}{2}\left(q_{e0}\boldsymbol{I}_{3}+\boldsymbol{q}^{\times}_{ev}\right)\in\mathbb{R}^{3\times3}$, correspondingly.

Subsequently, we consider the system's dynamical model with event-trigger mechanism
and input saturation issue. Firstly, we consider the impact of the event-trigger mechanism. Let $\boldsymbol{\tau}\in\mathbb{R}^{3}$ be the calculated control law, and denoting the sample-and-hold control signal derived by event-trigger mechanism as $\boldsymbol{\tau}_{e}\in\mathbb{R}^{3}$, an error variable can be defined as $\boldsymbol{e}_{\tau} \triangleq \boldsymbol{\tau}$ - $\boldsymbol{\tau}_{e}\in\mathbb{R}^{3}$ accordingly, as mentioned in \cite{wu2018event}.
We further consider the saturation issue. Since $\boldsymbol{\tau}_{e}$ is the desired output signal for actuators, a saturation variable can be defined as $\Delta\boldsymbol{\tau} = \boldsymbol{\tau}_{e} - \boldsymbol{u}\in\mathbb{R}^{3}$.
Therefore, applying the definition of $\boldsymbol{\tau}$, $\boldsymbol{\tau}_{e}$, $\boldsymbol{e}_{\tau}$ and $\Delta\boldsymbol{\tau}$, the actual actuator output $\boldsymbol{u}$ can be given as $\boldsymbol{u} = \boldsymbol{\tau} - \boldsymbol{e}_{\tau} - \Delta\boldsymbol{\tau}$. This relationship can be further clarified in Figure \ref{signal}.
\begin{figure}[hbt!]
	\centering 
	\includegraphics[width=0.4\textwidth]{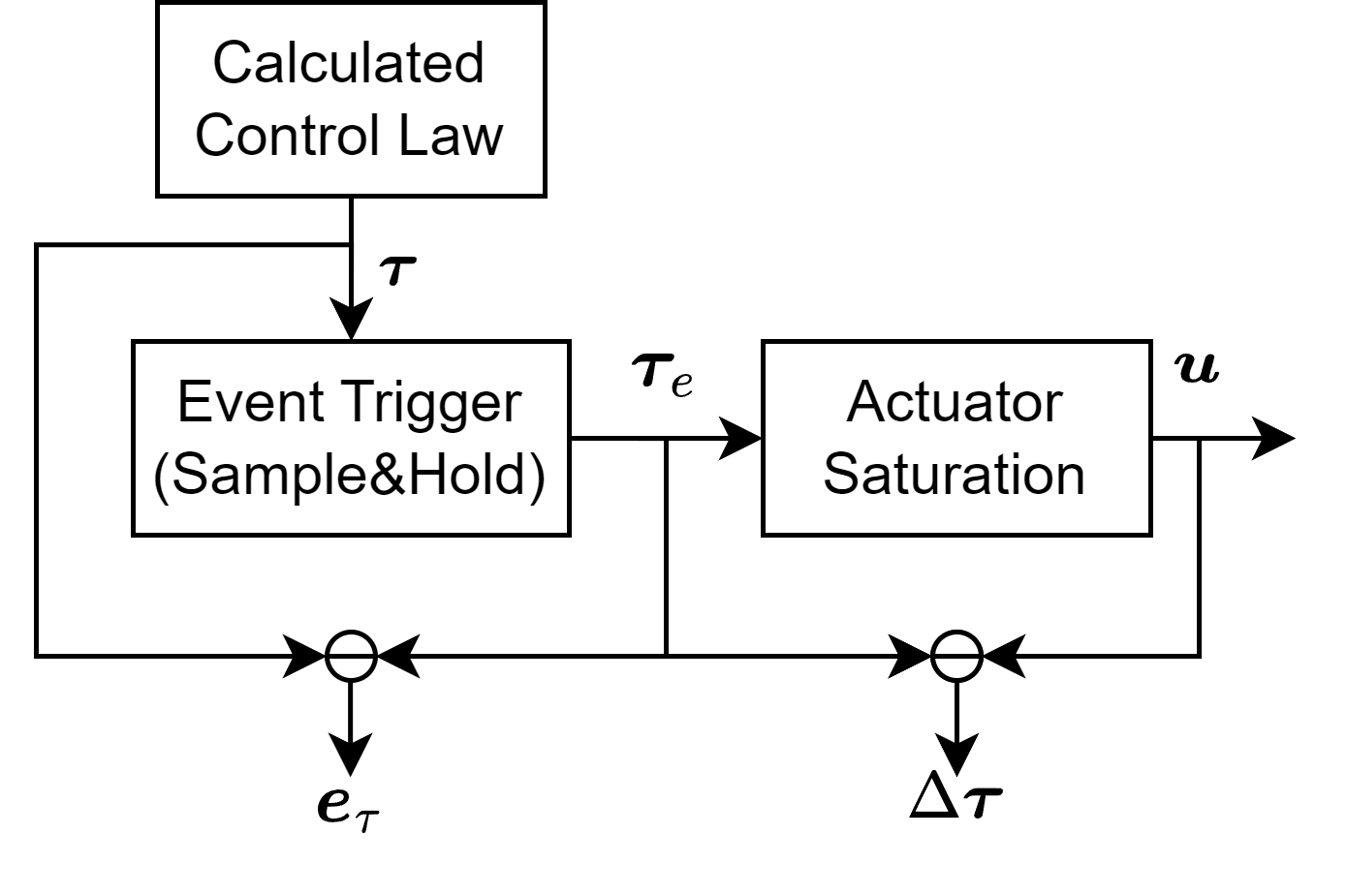}
	\caption{The Route of Control Signal}       
	\label{signal}   
\end{figure}

Accordingly, the dynamics equation (\ref{sys}) can be reformulated as follows:
\begin{equation}\label{sysetc}
	\boldsymbol{J}\dot{\boldsymbol{\omega}}_{e} = \boldsymbol{\Omega}_{e} + \boldsymbol{\tau}- \boldsymbol{e}_{\tau} - \Delta\boldsymbol{\tau} + \boldsymbol{d}
\end{equation}

The following assumptions are made for the synthesize of the controller:
\begin{assumption}\label{ASSJ}
	The inertia matrix $\boldsymbol{J}$ of the spacecraft is a known positive-definite matrix. Let $\lambda_{J\min}$, $\lambda_{J\max}$ be the minimum and the maximum of $\boldsymbol{J}$'s eigenvalue, we have $\lambda_{J\min}\|\boldsymbol{x}\|^{2}\le\boldsymbol{x}^{\text{T}}\boldsymbol{J}\boldsymbol{x} \le 	\lambda_{J\max}\|\boldsymbol{x}\|^{2}$ for all $\boldsymbol{x}\in\mathbb{R}^{3}$.
\end{assumption}
\begin{remark}
	Since the inertial uncertainty issue is not this paper's main consideration, we make the inertia matrix $\boldsymbol{J}$ a known one just for convenience. Notably, the presented framework can be extended to the condition that parameter uncertainty exists.
\end{remark}

\begin{assumption}\label{ASSD}
	The lumped disturbance $\boldsymbol{d}$ is unknown, but with a known upper boundary $D_{m}$, i.e., $\|\boldsymbol{d}\| \le D_{m}$ holds for all $t\in\left[0,+\infty\right)$.
\end{assumption}
\begin{assumption}\label{ASSGAMM}
	The vector part of the attitude error quaternion is non-zero for all $t\in\left[0,+\infty\right)$, i.e. $q_{e0}(t)\neq 0$ is satisfied all along.
\end{assumption}
Since $ q_{e0}\left(t\right) = 0$ means the total divergence of the attitude error system, which will actually not happen if the initial condition $ q_{e0}\left(0\right) \neq 0$ holds and the system is properly handled. Thus, such an assumption is reasonable.

\subsection{Constraint Description}\label{CONS}
\subsubsection{Attitude Rotation Rate Limitation}
For arbitrary upper bound constant given as $\Omega_{\max} > 0$, $\|\boldsymbol{\omega}_{s}\| \le \Omega_{\max}$ should be satisfied during the whole control process.

\subsubsection{Input Saturation Constraint}
Given arbitrary upper bound constant of maximum output value as $U_{\max} > 0$, for each $i(i = 1,2,3)$-th actuator axis, we have:
\begin{equation}
	u_{i} = 
	\begin{cases}
		\tau_{ei} &\quad, |\tau_{ei}| \le U_{\max}\\
		U_{\max}\text{sgn}(\tau_{ei})&\quad, |\tau_{ei}| > U_{\max}
	\end{cases}
\end{equation}
where $u_{i}$, $\tau_{ei}$ denotes the $i$-th component of $\boldsymbol{u}$ and $\boldsymbol{\tau}_{e}$, respectively. $\text{sgn}(\cdot)$ represents the symbolic function.

\subsection{Control Objective} 
The main control objective of this paper can be stated as follows: With the designed intermittent controller, the attitude error system given in equation \ref{sys} and \ref{sysetc} will be ultimately uniformly bounded, such that $\boldsymbol{q}_{ev}$ and $\boldsymbol{\omega}_{e}$ will all converge into a small region near the equilibrium point $\boldsymbol{q}_{ev} = \boldsymbol{0}$, $\boldsymbol{\omega}_{e} = \boldsymbol{0}$.

\section{MAIN DESIGN}
In this section, a composite event-triggered intermittent controller is presented to achieve the desired control objective. The proposed mechanism is comprised of two state-dependent trigger mechanisms that governs the activation and deactivation of the actuator. The stabilization of the system is then achieved through a layered strategy based on the backstepping control philosophy and the proposed trigger mechanism. Specifically, the total attitude system is firstly divided into a cascaded form, and then the composite event-trigger mechanism is applied to the dynamical layer to govern the actuator's behavior.

Following a backstepping control philosophy, we first denote the virtual control law that stabilizes the output layer $\boldsymbol{q}_{ev}$-system as $\boldsymbol{\omega}_{v}\in\mathbb{R}^{3}$. Further, we define a tracking-error subsystem $\boldsymbol{z}_{2}$ as $\boldsymbol{z}_{2} \triangleq \boldsymbol{\omega}_{e} - \boldsymbol{\omega}_{v}$. Accordingly, the proposed composite event-trigger mechanism will be applied to the $\boldsymbol{z}_{2}$-system.

\subsection{Composite Trigger Mechanism Design}

\subsubsection{The Turn-off Trigger Mechanism Design}

When the actuator is turned-on, a turn-off trigger mechanism is designed to judge when to turn off the actuator.
The turn-off trigger mechanism will shut down the actuator when the current sample-and-hold control signal is not able to guarantee the exponential convergence of the $\boldsymbol{z}_{2}$-system. Following such an idea \cite{wang2019event}, the turn-off trigger mechanism is designed as follows:
\begin{equation}\label{turnoff}
	\begin{aligned}
		t^{\text{act}}_{k} &= \inf_{t>t^{\text{on}}_{k}}\left\{\|\boldsymbol{e}_{\tau}\|^{2} > ae^{-\beta t} + b\right\}\\
		t^{\text{pas}}_{k} &= t^{\text{on}}_{k} + T_{\text{max}}\\
		t^{\text{off}}_{k} &= \min\left(t^{\text{act}}_{k},t^{\text{pas}}_{k}\right)
	\end{aligned}
\end{equation}
Here $t^{\text{on}}_{k}$, $t^{\text{off}}_{k}$ represents the arbitrary $k$-th time instant of the actuator-turn-on event and the actuator-turn-off event, respectively, $a, b, \beta > 0$ are positive design parameters, $T_{\max} > 0$ represents a given  maximum allowed opening-time for actuators. 

$t^{\text{act}}_{k}$, $t^{\text{pas}}_{k}$ represent two different time instants that govern the actual turn-off time instant $t^{\text{off}}_{k}$. It can be observed that when $\min\left(t^{\text{act}}_{k} , t^{\text{pas}}_{k}\right) = t^{\text{act}}_{k}$ holds, the actuator will be turned off due the the divergence of $\|\boldsymbol{e}_{\tau}\|^{2}$, while the actuator will be turned off for actuator protection consideration when  $\min\left(t^{\text{act}}_{k} , t^{\text{pas}}_{k}\right) = t^{\text{pas}}_{k}$ holds. With the designed turn-off trigger mechanism, for $t\in\left[t^{\text{on}}_{k},t^{\text{off}}_{k}\right]$, we have $\|\boldsymbol{e}_{\tau}\|^{2}\le ae^{-\beta t}+b$.

\subsubsection{The Turn-on Trigger Mechanism Design}

When the actuator is turned off, a turn-on trigger mechanism is designed to determine the condition for turning on the actuator.
The turn-on trigger mechanism is designed based on the purpose to establish a maximum allowed divergence upper boundary of $\boldsymbol{z}_{2}$-system, such that the temporarily divergence of $\boldsymbol{z}_{2}$-system (i.e., the dynamical layer subsystem) will not break up the exponential convergence of the total system.
The turn-on trigger mechanism is designed as follows:
\begin{equation}\label{turnon}
	t^{\text{on}}_{k+1} = \inf_{t > t^{\text{off}}_{k}}\left\{\|\boldsymbol{z}_{2}\|^{2} > \left(\rho_{0}-\rho_{\infty}\right)e^{-\gamma t} + \rho_{\infty}\right\}
\end{equation} 
where $\rho_{0}, \rho_{\infty}, \gamma > 0$ are positive design parameters that need indicating. Specifically, $\rho_{0} > 0$ stands for an initial value, $\rho_{\infty}$ stands for the value of the terminal asymptote, while $\gamma > 0$ governs the convergence rate. It can be observed that the designed turn-on trigger mechanism ensures that the time-evolution of $\|\boldsymbol{z}_{2}\|^{2}$ will be strictly restricted beneath an exponentially-converged function, of which the residual set is $\rho_{\infty}$. With the designed turn-on trigger mechanism, note that $\|\boldsymbol{z}_{2}\|^{2} \le (\rho_{0}-\rho_{\infty})e^{-\gamma t} + \rho_{\infty}$ holds for $t\in\left[t^{\text{off}}_{k},t^{\text{on}}_{k+1}\right]$.

The internal logic of the composite trigger mechanism is summarized in Figure \ref{LOGIC}. When the actuator is remain turned-on, the turn-off mechanism given in equation (\ref{turnoff}) is used to judge whether the actuator should be turned on or should be shut down. Similarly, the turn-on mechanism given in equation (\ref{turnon}) will be responsible for the judgment, determining whether the actuator should be turned off or should be turned on.
\begin{figure}[hbt!]
	\centering 
	\includegraphics[width=0.5\textwidth]{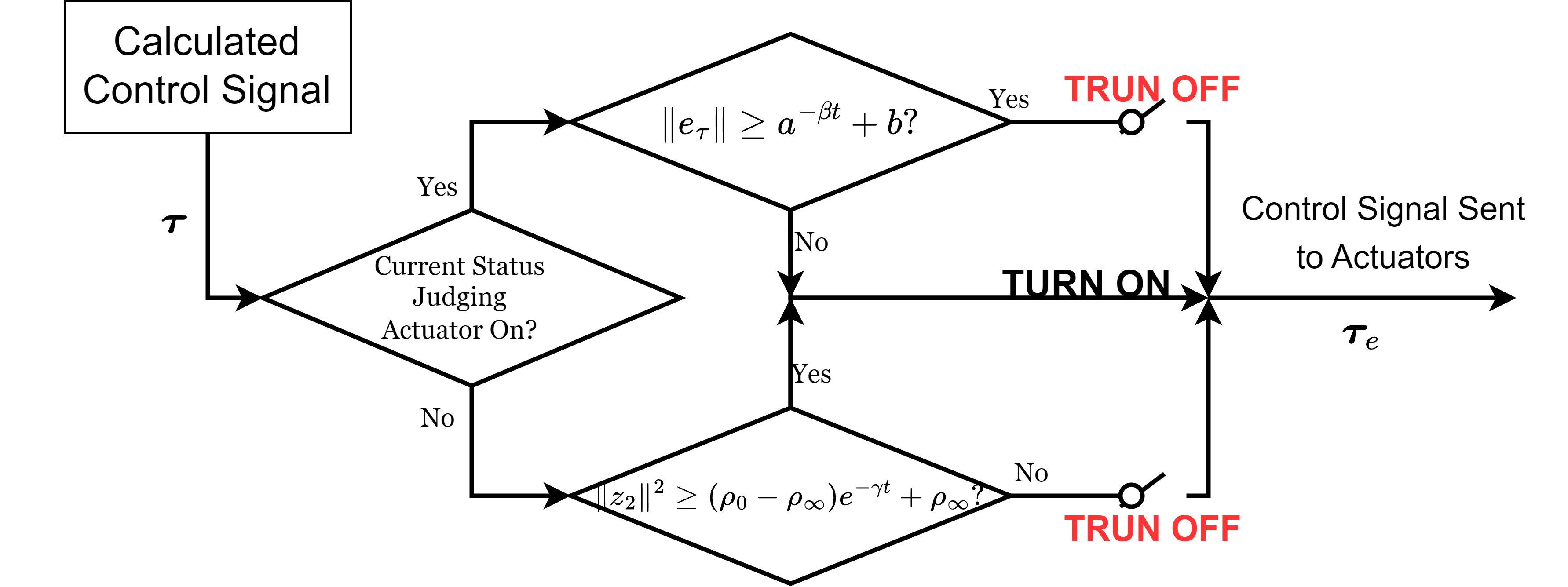}
	\caption{Internal Logic of the Composite Trigger Mechanism}       
	\label{LOGIC}   
\end{figure}

The detailed proof of the Minimum Inter-Event Time (MIET) of the proposed composite trigger mechanism is given in \cite{lei2022event}, and we omit it here for brevity. Briefly speaking, the possibility of Zeno behavior is ruled out from two perspective. 
On the one hand, since these two trigger mechanism is designed based on two variable $\|\boldsymbol{e}_{\tau}\|^{2}$ and $\|\boldsymbol{z}_{2}\|^{2}$, thus it is state-dependent, indicating that they all restricted by the system's dynamic. Owing to the fact that system's dynamic will not tend to be infinity through a properly designed controller, this guarantees the impossibility of the Zeno behavior. On the other hand, the non-zero constant $b$ and $\rho_{\infty}$ sets additional region for the time evolution of $\|\boldsymbol{e}_{\tau}\|^{2}$ and $\|\boldsymbol{z}_{2}\|^{2}$ to increase, hence there's no possibility that the MIET tend to be zero, which rules out the possibility of the Zeno behavior.

\subsection{Control Law Derivation}
Considering the error subsystem $\boldsymbol{q}_{ev}$, $\boldsymbol{z}_{2}$, the virtual control law $\boldsymbol{\omega}_{v}$ that stabilizes the $\boldsymbol{q}_{ev}$-system is given as follows:
\begin{equation}\label{virtual}
	\boldsymbol{\omega}_{v} = -\frac{|q_{e0}|M_{\omega}}{2\sqrt{3}}\boldsymbol{\varGamma}^{-1}_{e}\text{psat}\left(\boldsymbol{q}_{ev}\right)
\end{equation}
where $M_{\omega}>0$ is a design parameter that should be selected regarding to the attitude rotation rate limitation $\Omega_{\max}$, $\text{psat}(\boldsymbol{q}_{ev})\in\mathbb{R}^{3}$ is a column vector calculated by $\text{psat}(q_{evi})$ element-wisely, where $\text{psat}(\cdot)$ is a designed piece-wise function with saturation characteristic, expressed as follows:
\begin{equation}
	\begin{aligned}
		\text{psat}\left(q_{evi}\right)=
		\begin{cases}
			-a_{p}(q_{evi}+1)^{2}-1\quad &q_{evi}\in\left[-1,-P_{b}\right]\\
			K_{m}q_{evi}\quad &q_{evi}\in\left[-P_{b},+P_{b}\right]\\
			a_{p}\left(q_{evi}-1\right)^{2} + 1\quad &q_{evi}\in\left[P_{b},1\right]\\
		\end{cases}
	\end{aligned}
\end{equation}
where $P_{b}\in\left(0,1\right)$ is a parameter that needs indicating, standing for the segment point. According to the smooth and continuous condition, the parameter is then given as $a_{p} = \frac{1}{P^{2}_{b}-1}$, $K_{m} = \frac{-a_{p}(P_{b}-1)^{2}}{P_{b}}$. 
For instance, given $P_{b} = 0.6$, the graph of $\text{psat}(\cdot)$ is illustrated in Figure \ref{fig_p}.
\begin{figure}[hbt!]
	\centering 
	\includegraphics[scale = 0.06]{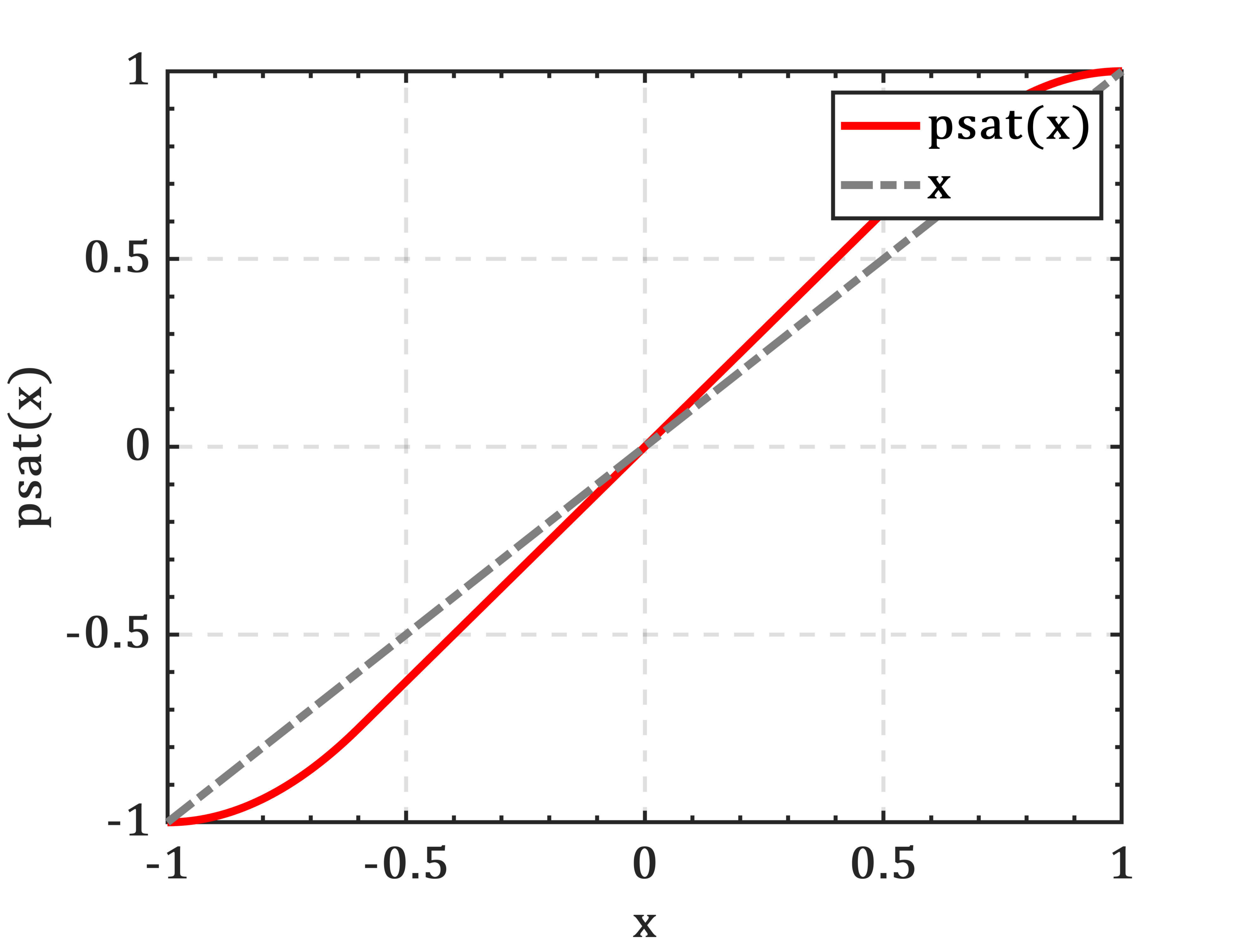}
	\caption{The function graph of $\text{psat}\left(\cdot\right)$ with $P_{b} = 0.6$}    
	\label{fig_p}  
\end{figure}

\begin{remark}
	Considering the norm of the designed $\boldsymbol{\omega}_{v}$ in equation (\ref{virtual}). Owing to the characteristic of the designed saturation function $\text{psat
	}\left(\cdot\right)$, we have $\|\text{psat}(\boldsymbol{q}_{ev})\|\le\sqrt{1+1+1}=\sqrt{3}$. Accordingly, it concludes that $\|\boldsymbol{\omega}_{v}\| \le \frac{|q_{e0}|M_{\omega}}{2\sqrt{3}}\|\boldsymbol{\varGamma}^{-1}_{e}\|\cdot\sqrt{3}$. Meanwhile, note that $\|\boldsymbol{\varGamma}^{-1}_{e}\| = \frac{2}{|q_{e0}|}$ holds, hence we have $\|\boldsymbol{\omega}_{v}\| \le M_{\omega}$. This indicates that the norm of $\boldsymbol{\omega}_{v}$ is strictly restricted by the inequality $\|\boldsymbol{\omega}_{v}\| \le M_{\omega}$. Therefore, by ensuring the actual $\boldsymbol{\omega}_{e}$ tracks the $\boldsymbol{\omega}_{v}$, the attitude rotation rate limitation can be achieved through a layered perspective.
\end{remark}
\begin{remark}
	The function $\text{psat}(\cdot)$ designed in this context acts as a nonlinear controller gain. It is important to note that virtual control laws are typically designed using an explicit controller gain parameter, such as $K_{q}\boldsymbol{q}_{ev}$. However, the use of an explicit parameter $K_{q}$ can violate the norm-bounded characteristic of $\|\boldsymbol{\omega}_{v}\|$. To overcome this limitation and ensure the system's convergence, a nonlinear piece-wise function $\text{psat}(\cdot)$ has been designed. By using this function, the system can exhibit a high-gain governed convergence behavior with limited virtual control law, and the effect of the function's high-gain behavior can be increased by decreasing the value of the parameter $P_{b}$.
\end{remark}
\begin{remark}
	Compared with existing saturated virtual control law method (e.g. \cite{li2016adaptive}), by choosing arbitrary $P_{b}\in\left(0,1\right)$, $|\text{psat}(x)| > x$ will be always holds. This alleviates the conservative of typically-used hyperbolic tangent function, i.e. $\tanh(C\cdot x)$, as $\tanh(Cx) > x$ can only satisfied when $C$ is sufficiently large.
\end{remark}


Subsequently, the actual control law $\boldsymbol{\tau}$ is designed as follows:
\begin{equation}\label{actual}
	\boldsymbol{\tau} = -\boldsymbol{\Omega}_{e}-K_{\omega}\boldsymbol{z}_{2} -\hat{\boldsymbol{d}} + \dot{\boldsymbol{\omega}}_{v} - K_{u}\boldsymbol{\xi} - \boldsymbol{P}_{q}
\end{equation}
where $K_{\omega}$, $K_{u} > 0$ are positive controller gain parameters, and $\boldsymbol{P}_{q}$ is defined as $\boldsymbol{P}_{q} \triangleq \boldsymbol{\varGamma}_{e}\boldsymbol{q}_{ev}$. $\hat{\boldsymbol{d}}$ is a rough compensation for disturbance, given as  $\hat{\boldsymbol{d}} \triangleq D_{m}\text{vec}\left(\tanh\frac{z_{2i}}{\mu_{i}}\right)(i=1,2,3)$, where $\mu_{i} > 0$ are design parameters, $z_{2i}$ represents the $i$-th component of $\boldsymbol{z}_{2}$. $\boldsymbol{\xi}\in\mathbb{R}^{3}$ is a compensation signal generated for the input saturation issue, governed by the following dynamical system \cite{bang2003large}:
\begin{equation}\label{comp}
	\dot{\boldsymbol{\xi}} = -\left[p_{1} + \frac{p_{2}\|\Delta\boldsymbol{\tau}\|^{2}}{\|\boldsymbol{\xi}\|^{2}}\right]\boldsymbol{\xi} + K_{\tau}\tanh\left(\Delta\boldsymbol{\tau}\right)
\end{equation}
where $p_{1}$, $p_{2} > 0$ are positive design parameters. $K_{\tau} > 0$ is the activation gain of the input saturation compensation system.

\section{STABILITY ANALYSIS}

\subsection{Main Theorem}
\begin{theorem}\label{T1}
	For the attitude error system given by equation (\ref{sys})(\ref{sysetc}), under the satisfaction of Assumption \ref{ASSJ}, \ref{ASSD} and \ref{ASSGAMM}, with the controller that given by (\ref{virtual})(\ref{actual})(\ref{comp}) and the composite event-trigger mechanism that given by equation (\ref{turnoff}) and (\ref{turnon}), the closed-loop
	system will be uniformly ultimately bounded (UUB), such that the error system will finally fall into a small residual set near the origin.
	
\end{theorem}
\begin{proof}
	The proof of Theorem \ref{T1} is given in Appendix \ref{STABLEPROOF}.
\end{proof}
The parameter selection suggestion is provided along with the stability analysis, given in Appendix \ref{PARA}.

\subsection{Further Discussion On the System's Behavior}\label{DISCUSS}

It is important to note that the turn-off trigger mechanism has been designed with a strict convergence guarantee for the $\boldsymbol{z}_{2}$ system in the time interval $\left[t^{\text{on}}_{k}, t^{\text{off}}_{k}\right]$. Consequently, the $\boldsymbol{z}_{2}$ system will decrease rapidly during this period, leading to an additional margin area between $V_{2}(t)$ and the upper boundary $\mathcal{J}_{2}(t)$. When the controller is turned off, $V_{2}(t)$ may diverge and eventually intersect with $\mathcal{J}_{2}(t)$, triggering the turn-on condition. This process occurs in a circular manner, and it can be deduced that the time response of $V_{2}(t)$ will remain below the $\mathcal{J}_{2}(t)$ determined by design parameters, resulting in a wavy curve-like time response of $V_{2}(t)$. This analysis will be further validated through simulation results in Section \ref{simulation}.

\section{SIMULATION AND ANALYSIS}\label{simulation}
This section presents several attitude control simulation results to validate the effectiveness of the proposed control scheme. 
Firstly, based on an established attitude tracking scenario, a normal case simulation is firstly carried out to show the fundamental ability of the proposed controller. Further, a comparison simulation is performed to provide specific analysis on the effect between the proposed intermittent controller and the typical periodic controller.

\subsection{Scenario Establishment}
In this section, the spacecraft is assumed to be a rigid-body spacecraft, of which the inertia matrix is assumed to be $\boldsymbol{J} = \left[2.8,0.002,0.0076;0.002,2.6,0.01;0.0076,0.01,1.9\right]kg\cdot m^{2}$. The maximum output torque is set to be $0.05N\cdot m$, while the allowed maximum continuously-output duration time $T_{\max} = 10s$.
Further, we utilize the mostly-considered periodically-varying external disturbance model, expressed as follows:
\begin{equation}
	\boldsymbol{d} = 
	\begin{bmatrix}
		1e-4 \cdot \left[4\sin\left(3\omega_{\text{dis}}t\right) + 2\cos\left(10\omega_{\text{dis}}t\right) -2\right]\\ 	
		1e-4\left[-1.5\sin\left(2\omega_{\text{dis}}t\right) + 3\cos\left(5\omega_{\text{dis}}t\right) +2\right]\\ 	
		1e-4\left[3\sin\left(10\omega_{\text{dis}}t\right) - 8\cos\left(4\omega_{\text{dis}}t\right) +2\right]\\ 	
	\end{bmatrix}
\end{equation}
where $\omega_{\text{dis}}$ represents the angular frequency of the disturbance, set to be $\omega_{\text{dis}} = 0.01\text{rad}/s$. Notably, such a disturbance is much bigger than the actual one in space environment, thus it is enough for the robustness evaluation.
Subsequently, the attitude rotation rate limitation of the spacecraft is required to not exceed $3^{\circ}/s$, indicating that $\|\boldsymbol{\omega}_{s}\| < 0.0524\text{rad}/s$ should be satisfied.

The simulation is performed at $10\text{Hz}$ and the duration time is $150s$. We set the judgment period for the trigger mechanism to be $\textbf{0.1s}$, while the control period is set to be $\textbf{1s}$. This indicates that the controller will be \textbf{only} opened on integer seconds, which is a common working period for actual spacecraft control systems.

\subsection{Normal Case Simulation: An Attitude Tracking Task}\label{NORMALCASE}
Firstly, we consider an attitude tracking control task. The initial condition of the spacecraft is given as $\boldsymbol{q}_{s}(0) = \left[0.4367,0.4927,0.5035,0.5595\right]^{\text{T}}$, $\boldsymbol{\omega}_{s} = \boldsymbol{0}$. 
The target attitude is given by the initial condition and the desired angular velocity, expressed as follows:
\begin{equation}
	\boldsymbol{q}_{d}(0) = [0,0,0,1]^{\text{T}};\boldsymbol{\omega}_{d} = 0.3\left[\cos\frac{T}{80},\sin\frac{T}{100},-\cos\frac{T}{100}\right]^{\text{T}}
\end{equation}
The unit is $^{\circ}/s$.
Notably, since $\boldsymbol{\omega}_{s} = \boldsymbol{\omega}_{e} + \boldsymbol{C}_{e}\boldsymbol{\omega}_{d}$ holds, thus $\|\boldsymbol{\omega}_{s}\|\le\|\boldsymbol{\omega}_{e}\|+\|\boldsymbol{\omega}_{d}\|$ holds.
Owing to the fact that $\|\boldsymbol{\omega}_{d}\|\neq 0$, we conservatively set $M_{\omega}$ as $M_{\omega} = 0.0175$.  
Other main parameters are given as $K_{\omega} = 5\cdot J_{i}(i = 1,2,3)$, $a = 0.1$, $b = 1e-4$, $\beta = 0.2$, $\rho_{0} = 10^{-3}$, $\rho_{\infty} = 10^{-5}$, $\gamma = 0.1$. Here $J_{i}$ denotes $\boldsymbol{J}$'s $i$ th  diagonal element.

The simulation result is illustrated in Figure \ref{fig_normaQE},\ref{fig_normaW},\ref{fig_normaU}. Figure \ref{fig_normaQE} shows the time evolution of $\boldsymbol{q}_{ev}$, while the time evolution of $\boldsymbol{\omega}_{s}$ is provided in Figure \ref{fig_normaW}. The actual exerted actuator output $\boldsymbol{u}$ is illustrated in Figure \ref{fig_normaU}. The time-responding of $V_{2}(t) = \frac{1}{2}\boldsymbol{z}^{\text{T}}_{2}\boldsymbol{J}\boldsymbol{z}_{2}$ and $\mathcal{J}_{2}(t)$ is further illustrated in Figure \ref{fig_normaS}.
\begin{figure}[hbt!]
	\centering 
	\includegraphics[width=0.5\textwidth]{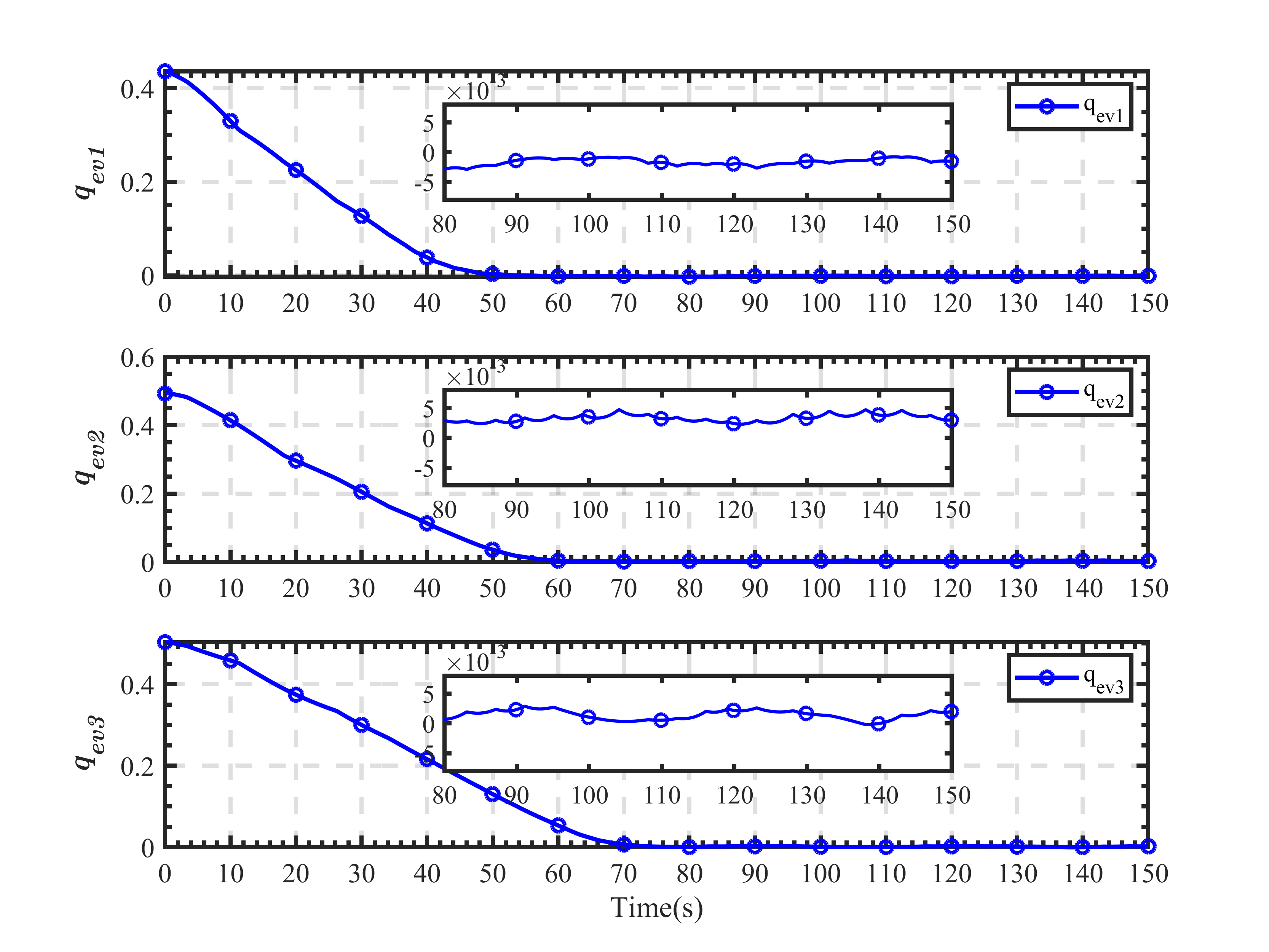}
	\caption{Time Evolution of the $i$-th component of $\boldsymbol{q}_{ev}\left(i=1,2,3\right)$ (Normal Case Simulation)}    
	\label{fig_normaQE}  
	\centering 
	\includegraphics[width=0.5\textwidth]{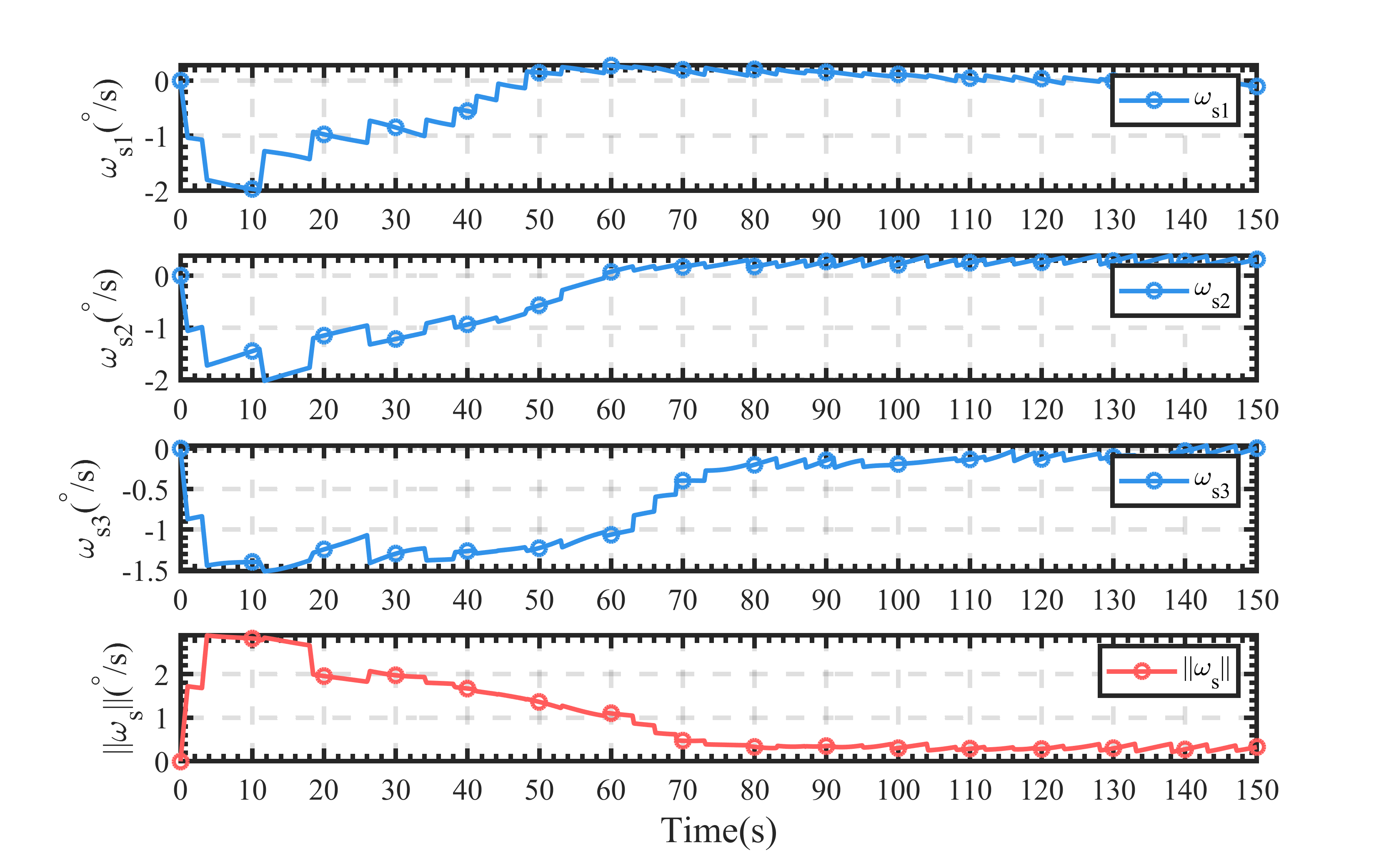}	\caption{Time Evolution of the $i$-th component of $\boldsymbol{\omega}_{s} \left(i=1,2,3\right)$ (Normal Case Simulation)}    
	\label{fig_normaW} 
	\centering 
	\includegraphics[width=0.5\textwidth]{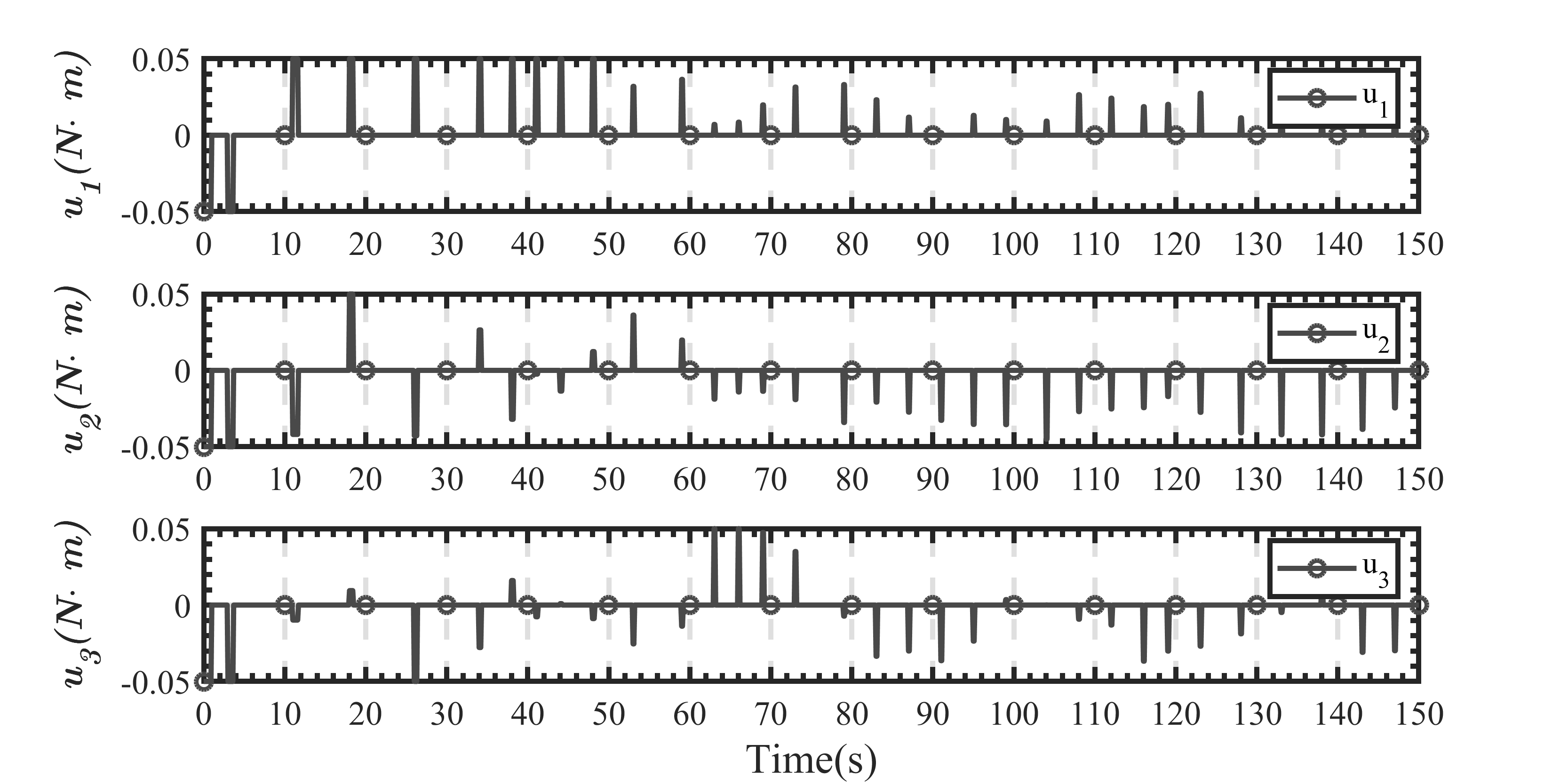}	\caption{Time Evolution of the $i$-th component of $\boldsymbol{u} \left(i=1,2,3\right)$ (Normal Case Simulation)}    
	\label{fig_normaU} 
\end{figure}
\begin{figure}[hbt!]
	\centering 
	\includegraphics[width=0.5\textwidth]{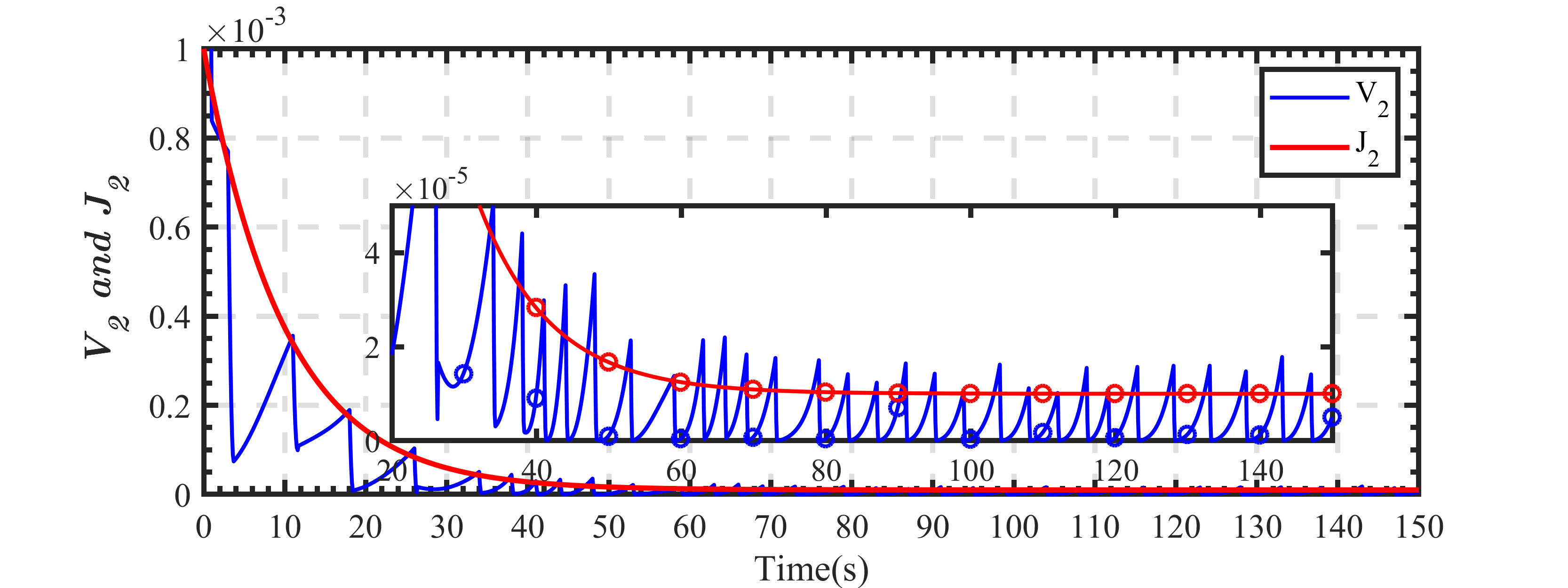}	\caption{Time Evolution of the $V_{2}(t)$ and $\mathcal{J}_{2}(t)$ (Normal Case Simulation, controller can only be opened in integer seconds)}    
	\label{fig_normaS} 
	\centering 
	\includegraphics[width=0.5\textwidth]{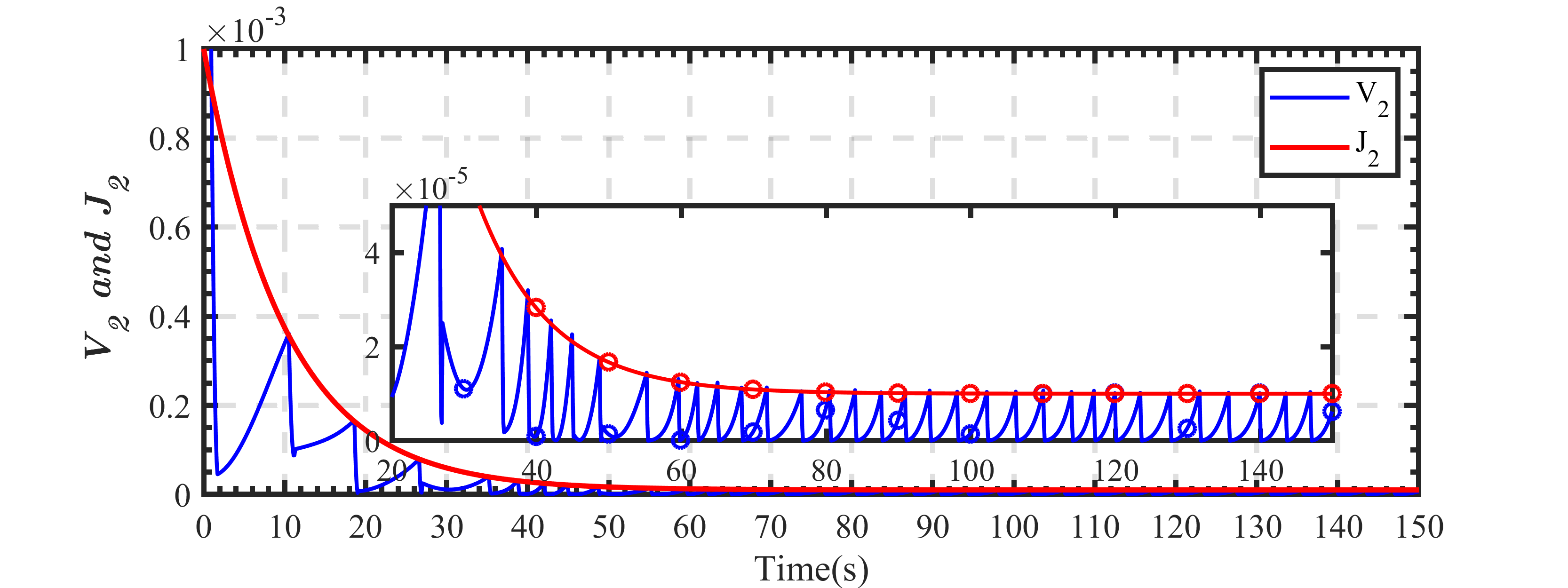}	\caption{Time Evolution of the $V_{2}(t)$ and $\mathcal{J}_{2}(t)$ (Additional Comparison, controller can be opened on every trigger judgment event)}    
	\label{fig_normaSS} 
\end{figure}

From Figure \ref{fig_normaU}, it can be observed that the actual actuator output shows an intermittent behavior. Each component of $\boldsymbol{q}_{ev}$ is able to converge to the steady state, with an accuracy about to be $0.5^{\circ}$. The spacecraft's angular velocity $\boldsymbol{\omega}_{s}$ is also able to track the desired $\boldsymbol{\omega}_{d}$. 

Note from Figure \ref{fig_normaS}, it can be discovered that the time evolution of $V_{2}(t)$ is actually acts like a wavy curve, which validates our proposition in Subsection \ref{DISCUSS}. The time evolution of $V_{2}$ seems not strictly satisfy the turn-on trigger mechanism. This is because we set actuators can only be opened in integer seconds, hence the actuator cannot be turned on at arbitrary time instant at which the condition is triggered. As an additionally elaboration, suppose that we allow the actuator to be turned on at any trigger moment, the corresponding time evolution of $V_{2}$ and $\mathcal{J}_{2}$ is illustrated in Figure \ref{fig_normaSS}.

\subsection{Comparison with Periodic Controller}
Further, we employ a "periodic" version of the proposed controller to perform a comparison simulation, of which the control frequency is set to be $1Hz$. The simulation scenario is the same as the one in Subsection \ref{NORMALCASE}, and the main control parameter of the periodic controller is regulated again to achieve a better performance, given as $K_{\omega} = 1.5$.
\begin{figure}[hbt!]
	\centering 
	\includegraphics[width=0.5\textwidth]{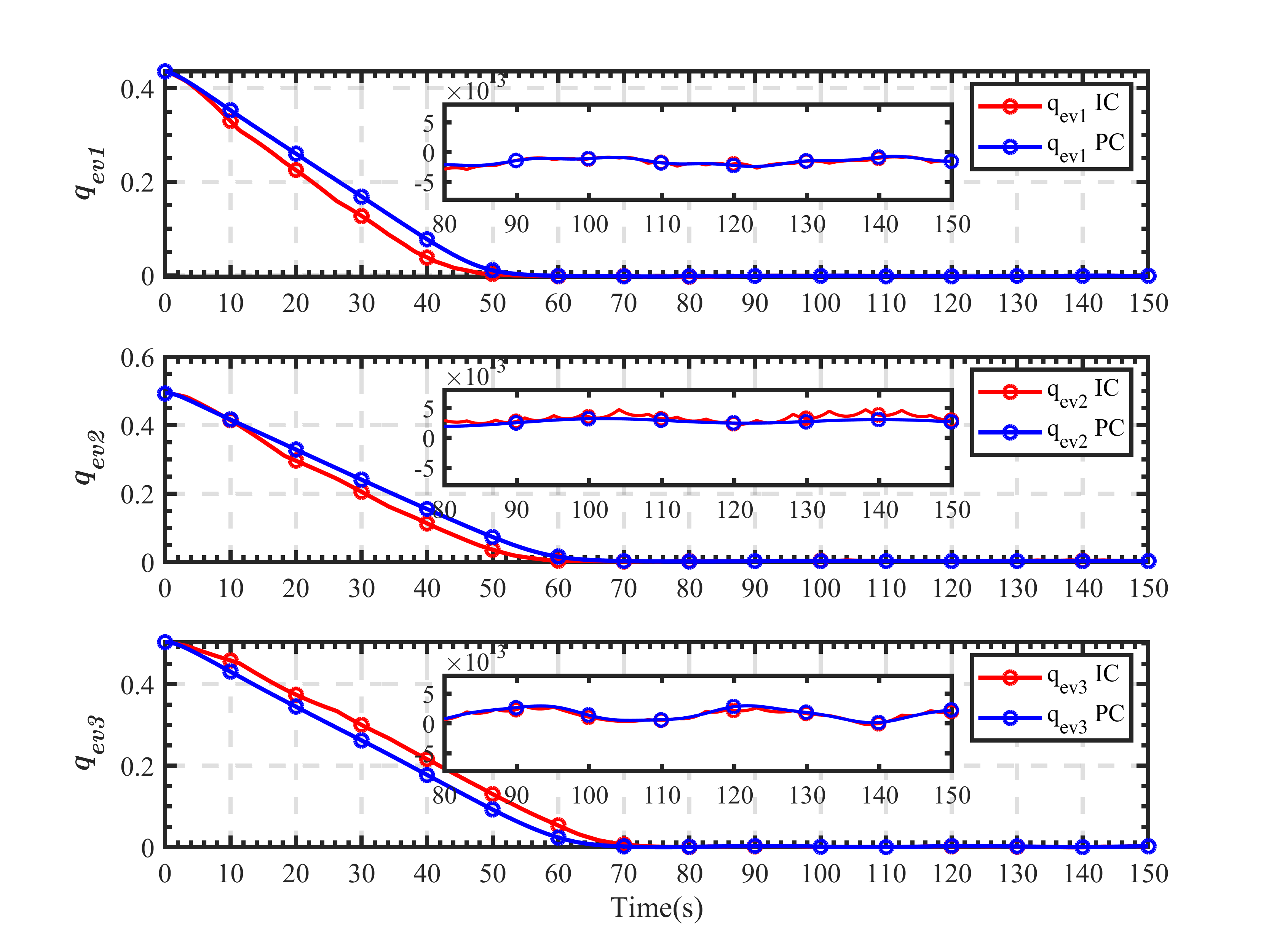}	\caption{Time Evolution of the $\boldsymbol{q}_{ev}(t)$ (Comparison Result of the Proposed Intermittent Controller (IC) and the benchmark Periodic Controller (PC) with $1Hz$ acting frequency)} 
	\label{fig_normaTC} 
\end{figure}
\begin{figure}[hbt!]
	\centering 
	\includegraphics[width=0.5\textwidth]{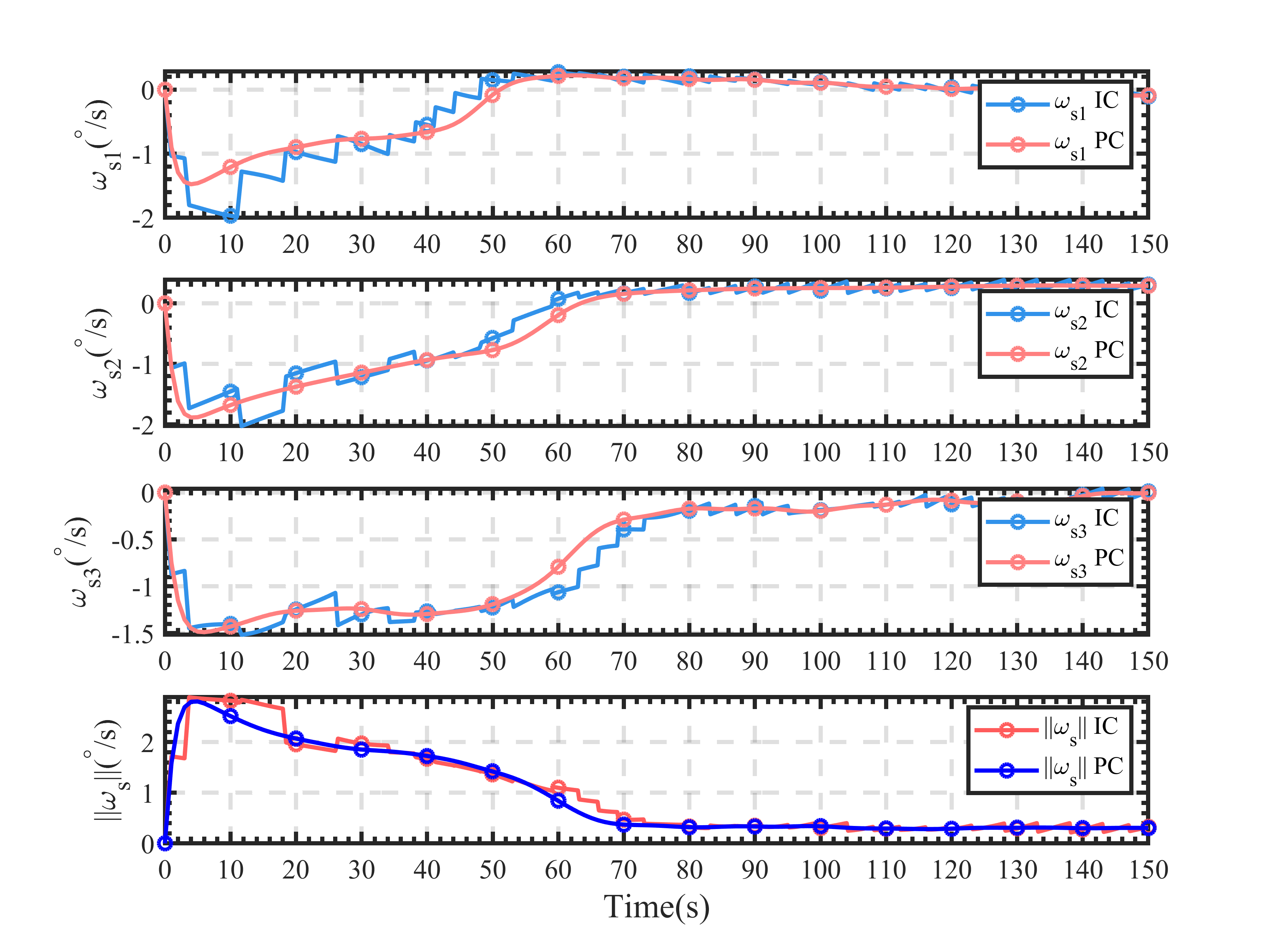}	\caption{Time Evolution of the $\boldsymbol{\omega}_{s}(t)$ (Comparison Result of the Proposed Intermittent Controller (IC) and the benchmark Periodic Controller (PC) with $1Hz$ acting frequency)} 
	\label{fig_WTC} 
\end{figure}
The comparison simulation result is illustrated in Figure \ref{fig_normaTC} and \ref{fig_WTC}, showing the time evolution of $\boldsymbol{q}_{ev}$ and $\boldsymbol{\omega}_{s}$ with the Intermittent Controller (proposed, IC) and the benchmark Periodic Controller (PC), respectively.

From the simulation that presented above, it can be inferred that the periodic controller does not exhibit a significantly higher level of control accuracy in comparison to the results obtained in Subsection \ref{NORMALCASE}, despite having a higher control frequency.
The total number of actuator actions performed by periodic controller is 150, whereas the intermittent controller proposed in this study exhibits only 32 instances of actuator action. This observation indicates that the proposed scheme is capable of achieving almost a same control accuracy with significantly fewer actuator actions.

\section{CONCLUSION}
This paper proposes a framework for achieving desired attitude control with intermittent actuator activation, while also addressing limitations in attitude rotation rate and input saturation. The proposed framework is based on a composite event-trigger mechanism, which consists of two state-dependent triggers responsible for activating and deactivating the actuators. By leveraging the backstepping philosophy and the composite trigger mechanism, an intermittent control framework is presented that stabilizes the system through a layered strategy. The proposed controller achieves a control accuracy of $0.5^{\circ}$ with a significantly-reduced actuator activation, and such an accuracy is comparable to that of a conventional periodic controller. This is particularly important since the steady-state maintenance phase is a critical part of a spacecraft's on-orbit scenario. The proposed method offers a potential way to achieve steady-state maintenance with much lower control frequency, making it a potential approach for practical applications.

For the further investigation, we may focus how to integrate such a framework with typical continuous or discrete implemented controller frameworks, and further building a hybrid control strategy that is able to switch automatically regarding the task's specific necessary.

\section{APPENDIX: Proof of Theorem 1}\label{STABLEPROOF}
The proof of Theorem \ref{T1} is divided into three steps, detailed in Subsection \ref{SECTURNON}, \ref{SECTURNOFF} and \ref{SECWHOLETIME}, respectively.
\subsection{Proof of Local Boundedness of $V(t)$ When Actuators are Turned On}\label{SECTURNON}

Firstly, we discuss on the closed-loop system' s responding behavior when actuators are turned on. Accordingly, we consider the time interval from the $k$-th turn-on trigger event to the $k$-th turn-off trigger event, i.e., $t\in\left[t^{\text{on}}_{k} , t^{\text{off}}_{k}\right]$.

Choosing a candidate Lyapunov function for the $\boldsymbol{q}_{ev}$-system as $V_{1} = \frac{1}{2}\boldsymbol{q}^{\text{T}}_{ev}\boldsymbol{q}_{ev}$. Taking the time-derivative of $V_{1}$ and combining with the attitude error system model given in equation (\ref{sys}), it can be yielded that:
\begin{equation}\label{dV10}
	\dot{V}_{1} = \boldsymbol{q}^{\text{T}}_{ev}\boldsymbol{\varGamma}_{e}\left(\boldsymbol{z}_{2}+\boldsymbol{\omega}_{v}\right) 
\end{equation}
Let $B_{\omega} \triangleq \frac{M_{\omega}}{\sqrt{3}}$ for brevity. Substituting the virtual control law $\boldsymbol{\omega}_{v}$ given in equation (\ref{virtual}) into equation (\ref{dV10}), it can be yielded that:
\begin{equation}\label{dV1}
	\begin{aligned}
		\dot{V}_{1} &= -\frac{|q_{e0}|M_{\omega}}{2\sqrt{3}}\boldsymbol{q}^{\text{T}}_{ev}\text{psat}(\boldsymbol{q}_{ev}) + \boldsymbol{q}^{\text{T}}_{ev}\boldsymbol{\varGamma}_{e}\boldsymbol{z}_{2}\\
		&\le -\frac{|q_{e0}|B_{\omega}}{2}\boldsymbol{q}^{\text{T}}_{ev}\boldsymbol{q}_{ev} + \boldsymbol{q}^{\text{T}}_{ev}\boldsymbol{\varGamma}_{e}\boldsymbol{z}_{2}\\ 
	\end{aligned}
\end{equation}

Subsequently, considering a candidate Lyapunov Function for the $\boldsymbol{z}_{2}$-system as $V_{2} = \frac{1}{2}\boldsymbol{z}^{\text{T}}_{2}\boldsymbol{J}\boldsymbol{z}_{2}$. Taking the time-derivative of $V_{2}$ and combining with the actual control law $\boldsymbol{\tau}$ given in equation (\ref{actual}), it can be obtained that:
\begin{equation}\label{dV21}
	\begin{aligned}
		\dot{V}_{2} &= \boldsymbol{z}^{\text{T}}_{2}\left[-K_{\omega}\boldsymbol{z}_{2}-K_{u}\boldsymbol{\xi}+\tilde{\boldsymbol{d}}-\boldsymbol{P}_{q} - \boldsymbol{e}_{\tau} - \Delta\boldsymbol{\tau}\right]
	\end{aligned}
\end{equation}
where $\tilde{\boldsymbol{d}}$ is defined as $\tilde{\boldsymbol{d}} \triangleq \boldsymbol{d} - D_{m}\text{vec}\left(\tanh\frac{z_{2i}}{\mu_{i}}\right)$. From equation (\ref{dV21}), it can be further yielded that:
\begin{equation}\label{dV2}
	\begin{aligned}
		\dot{V}_{2}
		&\le -K_{\omega}\|\boldsymbol{z}_{2}\|^{2}+\boldsymbol{z}^{\text{T}}_{2}\tilde{\boldsymbol{d}}+K_{u}\|\boldsymbol{z}_{2}\|\|\boldsymbol{\xi}\|+\|\boldsymbol{z}_{2}\|\|\boldsymbol{e}_{\tau}\|\\
		&\quad +\|\boldsymbol{z}_{2}\|\|\Delta\boldsymbol{\tau}\| - \boldsymbol{z}^{\text{T}}_{2}\boldsymbol{P}_{q}
	\end{aligned}
\end{equation}
Applying the Peter-Paul's inequality and the Young's inequality, it can be further derived that:
\begin{equation}\begin{aligned}\label{dV22}
		\dot{V}_{2} &\le -\left(K_{\omega}-\frac{b_{1}+K_{u}+1}{2}\right)\|\boldsymbol{z}_{2}\|^{2}+
		\frac{K_{u}}{2}\|\boldsymbol{\xi}\|^{2}\\
		&\quad +\frac{1}{2b_{1}}\|\boldsymbol{e}_{\tau}\|^{2}+\frac{1}{2}\|\Delta\boldsymbol{\tau}\|^{2}+\boldsymbol{z}^{\text{T}}_{2}\tilde{\boldsymbol{d}} - \boldsymbol{z}^{\text{T}}_{2}\boldsymbol{P}_{q}
	\end{aligned}
\end{equation}
where $b_{1} > 0$ is a positive coefficient.

Further, considering a candidate Lyapunov function for the $\boldsymbol{\xi}$-system as $V_{3} = \frac{1}{2}\boldsymbol{\xi}^{\text{T}}\boldsymbol{\xi}$. Taking the time-derivative of $V_{3}$, one can be obtained that:
\begin{equation}\begin{aligned}\label{dV3}
		\dot{V}_{3} &= -\left[p_{1} + \frac{p_{2}\|\Delta\boldsymbol{\tau}\|^{2}}{\|\boldsymbol{\xi}\|^{2}}\right]\boldsymbol{\xi}^{\text{T}}\boldsymbol{\xi} + \boldsymbol{\xi}^{\text{T}}K_{\tau}\tanh\left(\Delta\boldsymbol{\tau}\right)\\
		&\le -p_{1}\|\boldsymbol{\xi}\|^{2}-p_{2}\|\Delta\boldsymbol{\tau}\|^{2}+\frac{K_{\tau}}{2}\left[\|\boldsymbol{\xi}\|^{2}+\|\Delta\boldsymbol{\tau}\|^{2}\right]\\
		& = -\left(p_{1}-\frac{K_{\tau}}{2}\right)\|\boldsymbol{\xi}\|^{2} - \left(p_{2}-\frac{K_{\tau}}{2}\right)\|\Delta\boldsymbol{\tau}\|^{2}
	\end{aligned}
\end{equation}
Combining the result in equation (\ref{dV22}) and (\ref{dV3}), it can be further yielded that:
\begin{equation}\begin{aligned}\label{d23}
		\dot{V}_{2}+\dot{V}_{3}
		&\le-\left(K_{\omega
		}-\frac{b_{1}+K_{u}+1}{2}\right)\|\boldsymbol{z}_{2}\|^{2}\\
		&\quad -\left(p_{1}-\frac{K_{u}+K_{\tau}}{2}\right)\|\boldsymbol{\xi}\|^{2}\\
		&\quad -\left(p_{2}-\frac{K_{\tau}+1}{2}\right)\|\Delta\boldsymbol{\tau}\|^{2}+\frac{1}{2b_{1}}\|\boldsymbol{e}_{\tau}\|^{2}\\
		&\quad +\boldsymbol{z}_{2}^{\text{T}}\tilde{\boldsymbol{d}} - \boldsymbol{z}^{\text{T}}_{2}\boldsymbol{P}_{q}
	\end{aligned}
\end{equation}
Defining $C_{1}$, $C_{2}$ and $C_{3}$ as $C_{1} \triangleq K_{\omega}-\frac{b_{1}+K_{u}+1}{2}$, $C_{2} \triangleq p_{1}-\frac{K_{u}+K_{\tau}}{2}$ and $C_{3} \triangleq p_{2}-\frac{K_{\tau}+1}{2}$ for the writing brevity.
Note that we have $\frac{1}{2}\lambda_{J\max}\|\boldsymbol{z}_{2}\|^{2} \ge V_{2}$, thus  $-C_{1}\|\boldsymbol{z}_{2}\|^{2} \le -\frac{2C_{1}V_{2}}{\lambda_{J\max}}$ will be satisfied. By choosing parameters and ensures that $C_{1}, C_{2}, C_{3} > 0$ holds, it can be further derived that:
\begin{equation}\begin{aligned}\label{dVV2}
		\dot{V}_{2} + \dot{V}_{3}
		&\le -\frac{2C_{1}}{\lambda_{J\max}}V_{2} - 2C_{2}V_{3}+\frac{1}{2b_{1}}\|\boldsymbol{e}_{\tau}\|^{2} + \boldsymbol{z}^{\text{T}}_{2}\tilde{\boldsymbol{d}}-\boldsymbol{z}^{\text{T}}_{2}\boldsymbol{P}_{q}
	\end{aligned}
\end{equation}
Considering the term expressed as $\boldsymbol{z}^{\text{T}}_{2}\tilde{\boldsymbol{d}}$. By utilizing the Lemma 1 given in \cite{lei2023singularity}, we have $\boldsymbol{z}^{\text{T}}_{2}\tilde{\boldsymbol{d}} \le \sum_{i=1}^{3}D_{m}\left[|z_{2i}| - \tanh\frac{z_{2i}}{\mu_{i}}\right] \le 0.2785\sum_{i=1}^{3}\mu_{i}$, and we further define $D_{0} \triangleq 0.2785\sum_{i=1}^{3}\mu_{i}$ for convenient.

Overall, choosing a lumped candidate Lyapunov function as $V = V_{1}+V_{2}+V_{3}$. Taking the time-derivative of $V$ and combining with the result stated in equation (\ref{dV1}) and (\ref{dVV2}), it can be obtained that:
\begin{equation}\label{dvv}
	\begin{aligned}
		\dot{V} \le -C_{t}V + \frac{1}{2b_{1}}\|\boldsymbol{e}_{\tau}\|^{2}+D_{0}
	\end{aligned}
\end{equation}
where $C_{t}$ is a \textbf{constant}, defined as $C_{t} \triangleq \min\left(\min(|q_{e0}|)B_{\omega},\frac{2C_{1}}{\lambda_{J\max}},2C_{2}\right)$.
\begin{remark}\label{Ctremark}
	Notably, since the error system, i.e., $\boldsymbol{q}_{ev}$, will converge with properly designed controller, thus $|q_{e0}(t)|$ will increasing accordingly, hence there always exists a positive minima of $|q_{e0}(t)|$ during the whole control process. Practically, for the specific evaluation of $C_{t}$, it can be approximated by the initial condition, i.e., $\min\left(|q_{e0}(t)|\right) \approx |q_{e0}(0)|$. On the other hand, there always exists a sufficient small positive constant $\delta$ such that $\delta \le \min\left(|q_{e0}(t)|\right)$ holds, which can be used for the evaluation of $C_{t}$.
\end{remark}

Owing to the designed turn-off trigger mechanism given in equation (\ref{turnoff}), $\|\boldsymbol{e}_{\tau}\|^{2} \le ae^{-\beta t} + b$ always holds when actuators are turned on. Substituting this relationship into equation (\ref{dvv}) yields:
\begin{equation}\label{dV}
	\dot{V} \le -C_{t}V + \frac{1}{2b_{1}}\left(ae^{-\beta t} + b\right) + D_{0}
\end{equation}

Further, we consider the time evolution of $V(t)$ based on the Grownwall Inequality \cite{ye2007generalized}.
Integrating on both sides of equation (\ref{dV}) from the last turn-on trigger time instant $t^{\text{on}}_{k}$ to the current time instant $t$, it can be obtained that:
\begin{equation}\label{Vjifentemp}
	\begin{aligned}
		V(t) &\le V(t^{\text{on}}_{k})e^{-C_{t}(t-t^{\text{on}}_{k})}+ \frac{ae^{-C_{t}t}}{2b_{1}}\int_{t^{\text{on}}_{k}}^{t}e^{(C_{t}-\beta)s}ds \\
		&\quad + \left(D_{0}+\frac{b}{2b_{1}}\right)e^{-C_{t}t}\int_{t^{\text{on}}_{k}}^{t}e^{C_{t}s}ds\\
	\end{aligned}
\end{equation}

We first consider the condition that $C_{t}\neq \beta $ holds (as $C_{t}=\beta$ is a special case), it can be derived that:
\begin{equation}\label{integV}
	\begin{aligned}
		V(t) &\le  V(t^{\text{on}}_{k})e^{-C_{t}(t-t^{\text{on}}_{k})}\\
		&\quad +\frac{ae^{-\beta t}}{2b_{1}(C_{t}-\beta)} -\frac{ae^{-\beta t^{\text{on}}_{k}}}{2b_{1}(C_{t}-\beta)}e^{-C_{t}(t - t^{\text{on}}_{k})}\\
		&\quad + \frac{1}{C_{t}}\left(D_{0}+\frac{b}{2b_{1}}\right)\left(1-e^{-C_{t}\left(t-t^{\text{on}}_{k}\right)}\right)\\
	\end{aligned}
\end{equation}
Let $M_{k}(t^{\text{on}}_{k}) \triangleq \frac{ae^{-\beta t^{\text{on}}_{k}}}{2b_{1}} $ and $N_{k} \triangleq D_{0} + \frac{b}{2b_{1}}$ for brevity, it can be further derived that:
\begin{equation}\begin{aligned}\label{Vjifen}
		V\left(t\right) &\le \left[V\left(t^{\text{on}}_{k}\right)-
		\frac{M_{k}(t^{\text{on}}_{k})}{C_{t}-\beta} -\frac{N_{k}}{C_{t}}\right]e^{-C_{t}\left(t-t^{\text{on}}_{k}\right)}\\
		&\quad +\frac{M_{k}(t^{\text{on}}_{k})e^{-\beta (t-t^{\text{on}}_{k})}}{C_{t}-\beta}+\frac{N_{k}}{C_{t}}\\
	\end{aligned}
\end{equation}

By comparing the $C_{t}$ and the designed $\beta$, the discussion of the time-evolution of $V(t)$ is divided into the following conditions, stated as follows:

\textbf{1.} For $C_{t} -\beta > 0$, it can be observed that $e^{-\beta (t-t^{\text{on}}_{k})} > e^{-C_{t}(t-t^{\text{on}}_{k})}$ holds. Therefore, from equation (\ref{Vjifen}), it can be further derived that:
\begin{equation}\label{Vcond1}
	V(t)< \left[V(t^{\text{on}}_{k})-\frac{N_{k}}{C_{t}}\right]e^{-\beta(t-t^{\text{on}}_{k})}+\frac{N_{k}}{C_{t}}
\end{equation}
\textbf{2.} We then consider the condition that $C_{t} =\beta$ holds. Note that $e^{C_{t}-\beta} \ge 1$ holds for $C_{t}-\beta \ge 0$, given any $C^{'}_{t} > \beta$ that satisfies the previous discussed condition \textbf{1}, i.e., $C^{'}_{t}-\beta>0$, we have:
\begin{equation}
	\int_{t^{\text{on}}_{k}}^{t}1ds< \int_{t^{\text{on}}_{k}}^{t}e^{(C^{'}_{t}-\beta)s}ds
\end{equation}

hence the time evolution of $V(t)$ in this special case will be lower than the upper boundary of $V(t)$ in condition $\textbf{1}$, which is given by equation (\ref{Vcond1}). Therefore, it can be concluded that the following relationship holds for $C_{t}=\beta$:
\begin{equation}\label{Vcond0}
	V(t)< \left[V(t^{\text{on}}_{k})-\frac{N_{k}}{C_{t}}\right]e^{-\beta(t-t^{\text{on}}_{k})}+\frac{N_{k}}{C_{t}}
\end{equation}
which is actually a same expression as in equation (\ref{Vcond1}).

\textbf{3.} For the circumstance that $C_{t} < \beta$ holds, we first rearrange the equation (\ref{Vjifen}) as follows:
\begin{equation}
	\begin{aligned}
		V(t) &\le	
		\left[V(t^{\text{on}}_{k})-\frac{N_{k}}{C_{t}}\right]e^{-C_{t}(t-t^{\text{on}}_{k})}+\frac{N_{k}}{C_{t}}\\
		&\quad +M(t^{\text{on}}_{k})\left[\frac{e^{-\beta(t-t^{\text{on}}_{k})} - e^{-C_{t}(t-t^{\text{on}}_{k})}}{C_{t}-\beta}\right]
	\end{aligned}
\end{equation}
Considering a binary function $E(r,q) \triangleq e^{-r\cdot q}(q,r>0)$. Taking the partial-derivative of $E(r,q)$ with respect to $r$, one can be obtained that $\frac{\partial E(r,q)}{\partial r} = -qe^{-r\cdot q} < 0$. Therefore, for $q= t-t^{\text{on}}_{k}$ and $r = C_{t}$ and $C_{t} < \beta$, we have:
\begin{equation}
	\begin{aligned}
		|-(t-t^{\text{on}}_{k})e^{-\beta(t-t^{\text{on}}_{k})}| &< |\frac{e^{-C_{t}(t-t^{\text{on}}_{k})}}{C_{t}-\beta} - \frac{e^{-\beta (t-t^{\text{on}}_{k})}}{C_{t}-\beta}|\\
		& < |-(t-t^{\text{on}}_{k})e^{-C_{t}(t-t^{\text{on}}_{k})}|
	\end{aligned}
\end{equation}
As stated in equation (\ref{turnoff}), the inter-event time between the $k$-th turn-on trigger time instant to the $k$-th turn-off trigger is restricted by the maximum allowed turn-on duration time $T_{\max}$, thus we have $t -t^{\text{on}}_{k} \le T_{\max}$. Accordingly, this yields the following result:
\begin{equation}\label{Vcond2}
	V(t) < \left[V(t^{\text{on}}_{k}) + M(t^{\text{on}}_{k})T_{\max}-\frac{N_{k}}{C_{t}}\right]e^{-C_{t}(t-t^{\text{on}}_{k})}+\frac{N_{k}}{C_{t}}
\end{equation}
Notably, recalling the definition of $M(t^{\text{on}}_{k}) = \frac{ae^{-\beta t^{\text{on}}_{k}}}{2b_{1}}$, we have $\lim_{k\to+\infty}M(t^{\text{on}}_{k}) = 0$.

The result in equation (\ref{Vcond1}) and (\ref{Vcond2}) indicates that the time evolution of $V(t)$ will beneath an exponentially-converged function. The right hand side of equation (\ref{Vcond1}) and (\ref{Vcond2}) can be regarded as an exponential function with an initial value $V(t^{\text{on}}_{k})$ for $C_{t}-\beta\ge0$ or $V(t^{\text{on}}_{k})+M_{k}(t^{\text{on}}_{k})T_{\max}$ for $C_{t}-\beta > 0$, and an exponential index $-\beta$ or $-C_{t}$, while the residual set can be uniformly given as follows:
\begin{equation}\label{RES}
	\lim_{k\to+\infty}\mathop{V}_{t>t^{\text{on}}_{k}}(t) < \frac{N_{k}}{C_{t}}
\end{equation}
For $t\in\left[t^{\text{on}}_{k},t^{\text{off}}_{k}\right]$, this completes the proof of the system's \textit{local boundedness} when actuators are turned on.
\begin{remark}
	The result in equation (\ref{RES}) hints that a larger $C_{t}$ will lead to a higher accuracy, or a small residual set equivalently. Oppositely, if $C_{t}$ is not big enough, recalling the definition of $N_{k}$, it can be noticed that a relatively small $b$ should be chosen to obtain a considerable performance.
\end{remark}

\subsection{Proof of Local Boundedness of $V(t)$ when Actuators are Turned Off}\label{SECTURNOFF}

Next, we discuss the system's behavior when actuators are turned off, focusing on the time interval between the $k$-th turn-off trigger time instant $t^{\text{off}}_{k}$ and the $k+1$-th turn-on trigger time instant $t^{\text{on}}_{k+1}$, i.e., $t\in\left[t^{\text{off}}_{k},t^{\text{on}}_{k+1}\right]$.

Considering the result for the $\boldsymbol{q}_{ev}$-system given in equation (\ref{dV1}), by applying the Peter-Paul 's inequality, it can be further obtained that:
\begin{equation}
	\begin{aligned}
		\dot{V}_{1} &\le -\frac{|q_{e0}|B_{\omega}}{2}\|\boldsymbol{q}_{ev}\|^{2} + \|\boldsymbol{z}_{2}\|\cdot\|\boldsymbol{P}_{q}\| \\
		&\le -\frac{|q_{e0}|B_{\omega}}{2}\|\boldsymbol{q}_{ev}\|^{2} + \frac{b_{2}}{2}\|\boldsymbol{q}_{ev}\|^{2} + \frac{1}{2b_{2}}\|\boldsymbol{z}_{2}\|^{2}\\
		&= -\left(|q_{e0}|B_{\omega} - b_{2}\right)V_{1} + \frac{1}{2b_{2}}\|\boldsymbol{z}_{2}\|^{2}
	\end{aligned}
\end{equation}
where $b_{2} > 0$ is a positive coefficient. 
For the writing convenience, let $m \triangleq \rho_{0}$ and $n \triangleq \rho_{\infty}(1-e^{-\gamma t})$ defined respectively.
Combining the result with the given actuator turn-on trigger mechanism (\ref{turnon}), we have:
\begin{equation}\label{dV11}
	\begin{aligned}
		\dot{V}_{1} \le -(|q_{e0}|B_{\omega}-b_{2})V_{1} + \frac{m}{2b_{2}}e^{-\gamma t} + \frac{n}{2b_{2}}
	\end{aligned}
\end{equation}
Further, considering the expression of $\dot{V}_{3}$, it can be derived that:
\begin{equation}\begin{aligned}\label{dV32}
		\dot{V}_{3} &\le -\left(p_{1}-\frac{K_{\tau}}{2}\right)\|\boldsymbol{\xi}\|^{2} - \left(p_{2}-\frac{K_{\tau}}{2}\right)\|\Delta\boldsymbol{\tau}\|^{2}
	\end{aligned}
\end{equation}
Recalling the definition of $C_{2} = p_{1} - \frac{K_{u} + K_{\tau}}{2}$, $C_{3} = p_{2} - \frac{K_{\tau}}{2}$ and the facts that parameters have been chosen to ensure that $C_{2}, C_{3} > 0$ holds, it can be further obtained that $\dot{V}_{3} \le -2C_{2}V_{3}$ holds.

Let $C_{d}$ be a design parameter-related constant, defined as $C_{d} \triangleq \min\left(\min(|q_{e0}|)B_{\omega} - b_{2}, 2C_{2}\right)$. Considering a combined Lyapunov function as $V_{n} = V_{1} + V_{3}$, we have:
\begin{equation}\label{dVn}
	\dot{V}_{n} \le -C_{d}V_{n} + \frac{m}{2b_{2}}e^{-\gamma t} + \frac{n}{2b_{2}}
\end{equation}

Similarly as we done previously, we first consider the condition that $C_{d}\neq\gamma$ holds. For $t>t^{\text{off}}_{k}$, integrating on both sides of equation (\ref{dVn}), it can be yielded that:
\begin{equation}
	\begin{aligned}\label{dV1_2}
		V_{n}(t) &\le V_{n}(t^{\text{off}}_{k})e^{-C_{d}(t-t^{\text{off}}_{k})}+ \frac{me^{-C_{d}t}}{2b_{2}}\int_{t^{\text{off}}_{k}}^{t}e^{(C_{d}-\gamma)s}ds\\
		&\quad + \frac{n}{2b_{2}}\int_{t^{\text{off}}_{k}}^{t}e^{-C_{d}\left(t - s\right)}ds\\
		&= V_{n}(t^{\text{off}}_{k})e^{-C_{d}(t-t^{\text{off}}_{k})} +\frac{me^{-\gamma t}}{2b_{2}(C_{d}-\gamma)} \\
		&\quad -\frac{me^{-\gamma t^{\text{off}}_{k}}}{2b_{2}(C_{d}-\gamma)}e^{-C_{d}(t - t^{\text{off}}_{k})}+ \frac{n}{2b_{2}C_{d}}\left(1-e^{-C_{d}\left(t-t^{\text{off}}_{k}\right)}\right)\\
	\end{aligned}
\end{equation}
We further define $M_{f}(t^{\text{off}}_{k}) \triangleq \frac{me^{-\gamma t^{\text{off}}_{k}}}{2b_{2}} $ and $N_{f} \triangleq \frac{n}{2b_{2}}$, hence equation (\ref{dV1_2}) can be further rearranged as follows:
\begin{equation}\label{Vn}
	\begin{aligned}
		V_{n}(t) &\le \left[V_{n}(t^{\text{off}}_{k}) - \frac{M_{f}(t^{\text{off}}_{k})}{C_{d}-\gamma}-\frac{N_{f}}{C_{d}}\right]e^{-C_{d}(t-t^{\text{off}}_{k})}\\
		&\quad +\frac{M_{f}(t^{\text{off}}_{k})e^{-\gamma (t-t^{\text{off}}_{k})}}{C_{d}-\gamma}+\frac{N_{f}}{C_{d}}\\
	\end{aligned}
\end{equation}

Further, We consider the time evolution of $V_{2}(t)$. As per Assumption \ref{ASSJ}, it can be obtained that $V_{2}\le\frac{\lambda_{J\max}}{2}\|\boldsymbol{z}_{2}\|^{2}$ holds. By utilizing the corollary derived from the turn-on trigger mechanism, we have:
\begin{equation}\label{VRHO}
	\begin{aligned}
		V_{2} \le \frac{\lambda_{J\max}}{2}\left[(\rho_{0}-\rho_{\infty})e^{-\gamma t}+\rho_{\infty}\right]
	\end{aligned}
\end{equation}
Defining the right side of equation (\ref{VRHO}) as a judgment function $\mathcal{J}_{2}(t)$, expressed as  $\mathcal{J}_{2}(t) \triangleq \frac{\lambda_{J\max}}{2}\left[(\rho_{0}-\rho_{\infty})e^{-\gamma t}+\rho_{\infty}\right]$. Accordingly, it can be yielded that:
\begin{equation}\label{J2}
	\begin{aligned}
		\mathcal{J}_{2}(t) &= \frac{\lambda_{J\max}e^{-\gamma t^{\text{off}}_{k}}}{2}(\rho_{0}-\rho_{\infty})e^{-\gamma(t-t^{\text{off}}_{k})}+\frac{\lambda_{J\max}\rho_{\infty}}{2}\\
		&= \left[\mathcal{J}_{2}(t^{\text{off}}_{k})-\frac{\lambda_{J\max}\rho_{\infty}}{2}\right]e^{-\gamma(t-t^{\text{off}}_{k})}+\frac{\lambda_{J\max}\rho_{\infty}}{2}\\
	\end{aligned}
\end{equation}

For the derivation of the final conclusion, considering a lumped Lyapunov function as $V_{a}(t) \triangleq V_{n}(t) + \mathcal{J}_{2}(t)$. Owing to the fact that $V_{2}(t)\le \mathcal{J}_{2}(t)$ holds, we have $V(t) \le V_{a}(t)$.

Combining with the result in equation (\ref{Vn}) and (\ref{J2}) and following a similar analysis in Subsection \ref{SECTURNON}, it can be concluded that:

\textbf{1.} For $C_{d}-\gamma > 0$, $e^{-C_{d}(t-t^{\text{off}}_{k})} < e^{-\gamma(t-t^{\text{off}}_{k})}$ will be satisfied. Correspondingly, we have:
\begin{equation}\label{Vcond5}
	\begin{aligned}
		V(t)&<\left[V_{a}(t^{\text{off}}_{k})-\frac{N_{f}}{C_{d}}-\frac{\lambda_{J\max}\rho_{\infty}}{2}\right]e^{-\gamma(t-t^{\text{off}}_{k})}\\
		&\quad +\frac{N_{f}}{C_{d}} +\frac{\lambda_{J\max}\rho_{\infty}}{2}
	\end{aligned}
\end{equation}
\textbf{2. } Similarly as we done previously, the result for the special case $C_{d}=\gamma$ can be directly derived as follows:
\begin{equation}\label{Vcond00}
	\begin{aligned}
		V(t)&<\left[V_{a}(t^{\text{off}}_{k})-\frac{N_{f}}{C_{d}}-\frac{\lambda_{J\max}\rho_{\infty}}{2}\right]e^{-\gamma(t-t^{\text{off}}_{k})} \\
		&\quad
		+\frac{N_{f}}{C_{d}} +\frac{\lambda_{J\max}\rho_{\infty}}{2}
	\end{aligned}
\end{equation}

\textbf{3.} For $C_{d} -\gamma < 0$, note that $e^{-\gamma(t-t^{\text{off}}_{k})} < e^{-C_{d}(t-t^{\text{off}}_{k})}$ holds, thus we have $\mathcal{J}_{2}(t) < \left[\mathcal{J}_{2}(t^{\text{off}}_{k}) - \frac{\lambda_{J\max}\rho_{\infty}}{2}\right]e^{-C_{d}(t-t^{\text{off}}_{k})} +\frac{\lambda_{J\max}\rho_{\infty}}{2}$. Combining these relationships, we have:
\begin{equation}\label{Vcond6}
	\begin{aligned}
		&V(t) \\
		<& \left[V_{a}(t^{\text{off}}_{k}) +  M_{f}(t^{\text{off}}_{k})T_{\text{maxoff}}-\frac{N_{f}}{C_{d}}-\frac{\lambda_{J\max}\rho_{\infty}}{2}\right]e^{-C_{d}(t-t^{\text{off}}_{k})} \\
		&\quad +\frac{N_{f}}{C_{d}} +\frac{\lambda_{J\max}\rho_{\infty}}{2}
	\end{aligned}
\end{equation}
where $T_{\text{maxoff}}$ is the maxima of the inter-event time between the $k$-th turn-off trigger to the $k+1$-th turn-on trigger. Owing to the fact that $\boldsymbol{z}_{2}$-system is perturbed by various perturbations, such as the time-derivative of $\boldsymbol{\omega}_{v}$ and the external disturbance, a maximum value $T_{\text{maxoff}}$ exists such that $T_{\text{maxoff}} \ge \max(t-t^{\text{off}}_{k})$ holds for all possible circumstances.

Consequently, owing to the designed turn-on trigger mechanism (\ref{turnon}), the time evolution of the lumped Lyapunov function $V(t)$ will be maintained beneath an exponentially converged function even when actuators are turned off, and the lumped system will finally fall into a uniformed region, expressed as follows:
\begin{equation}\label{RES1}
	\lim_{k\to\infty}\mathop{V}_{t>t^{\text{off}}_{k}}(t) < \frac{N_{f}}{C_{d}}+\frac{\lambda_{J\max}\rho_{\infty}}{2}
\end{equation}
Correspondingly, for $t\in\left[t^{\text{off}}_{k},t^{\text{on}}_{k+1}\right]$, this completes the proof of the system's \textit{local boundedness} when actuators are turned off.
\begin{remark}
	It can be inferred that the $\boldsymbol{z}_{2}$-system will not be able to converge exponentially under such a condition. However, it can be noticed intuitively that the designed turn-on mechanism (\ref{turnon}) will restrict such a divergence effect in a range such that it will not break up the whole system's stability.
\end{remark}

\subsection{Characterization of the Global Uniformly Continuous Upper Boundary of $V(t)$ on Whole Time domain}\label{SECWHOLETIME}

According to the analysis given in Subsection \ref{SECTURNON} and \ref{SECTURNOFF}, it can be observed that the given upper boundary of $V(t)$ shows different expression on each inter-event time, and these seperated upper boundary may even not connected with each other. Thus, the system's stability result cannot be directly yielded by those \textit{Local Boundedness} results. To facilitate the stability analysis on the entire time domain, this subsection further characterizes the system's behavior by introducing a property of the exponentially-converged function, thereby extending the \textit{locally boundedness} result that given by equation (\ref{Vcond1})(\ref{Vcond2})(\ref{Vcond5}) and (\ref{Vcond6}) to a global one.

The right hand side of equation (\ref{Vcond1})(\ref{Vcond2})(\ref{Vcond5}) and (\ref{Vcond6}) stands for an upper boundary of the time evolution of $V(t)$ on each inter-event time period under different conditions. Considering arbitrary $i$-th trigger time instant (including both turn-on trigger and the turn-off trigger), for all these four conditions, the upper boundary takes a similar exponentially-converged form, which can be characterized as follows:
\begin{equation}
	A^{i}(t) \triangleq (A^{i}(t^{i}_{\text{trig}}) - A^{i}_{\infty})e^{-L^{i}_{A}(t-t^{i}_{\text{trig}})} + A^{i}_{\infty}(t\in\left[t^{i}_{\text{trig}},t^{i+1}_{\text{trig}}\right])
\end{equation}
where $t^{i}_{\text{trig}}$ denotes the time instant of the counted $i$-th trigger. For instance, considering the expression given in equation (\ref{Vcond1}), it can be observed that $t^{i}_{\text{trig}} = t^{\text{on}}_{k}$, $t^{i+1}_{\text{trig}} = t^{\text{off}}_{k}$, $A(t^{i}_{\text{trigger}}) = V(t^{\text{on}}_{k})$, $A^{i}_{\infty} = \frac{N_{k}}{C_{t}}$ and $L^{i}_{A} = \beta$.

Note that for each $i$-th upper boundary function $A^{i}(t)$, we have the  following property:
\begin{equation}\label{AA}
	\begin{aligned}
		&(A^{i}(t^{i}_{\text{trig}}) - A^{i}_{\infty})e^{-L^{i}_{A}(t-t^{i}_{\text{trig}})} + A^{i}_{\infty}\\
		&\quad = (A^{i}(0)-A^{i}_{\infty})e^{-L^{i}_{A}t} +A^{i}_{\infty}
	\end{aligned}
\end{equation}
This indicates that the $i$-th upper boundary function $A^{i}(t)$ that defined on the inter-event time period $t\in\left[t^{i}_{\text{trig}},t^{i+1}_{\text{trig}}\right]$ is actually a part of an exponentially-converged function that defined on the whole time domain $t\in\left[0,+\infty\right)$. See Figure \ref{ex} for a further explanation. 
\begin{figure}[hbt!]
	\centering 
	\includegraphics[width=0.5\textwidth]{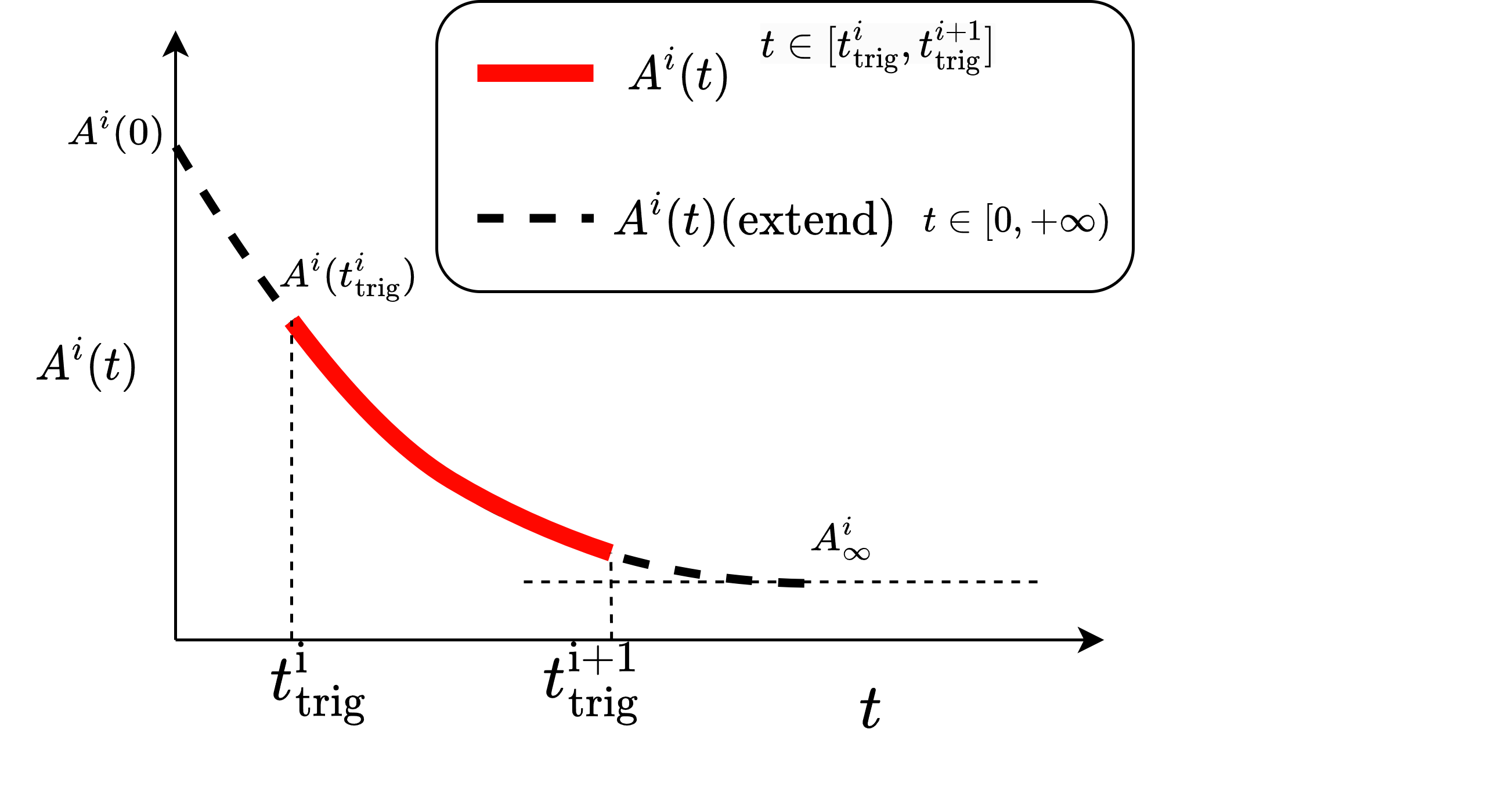}
	\caption{Extending of the local boundary $A^{i}(t)$ from $t\in\left[t^{i}_{\text{trig}},t^{i+1}_{\text{trig}}\right]$ to $t\in\left[0,+\infty\right)$}       
	\label{ex}   
\end{figure}
This allows us to extend the function segment to the whole time domain.

Regarding different conditions depicted in equation (\ref{Vcond1})(\ref{Vcond2})(\ref{Vcond5}) and (\ref{Vcond6}), for each $i$-th trigger time instant $t^{i}_{\text{trig}}$, four functions can be characterized accordingly, denoted as $A^{i}_{1}(t)$, $A^{i}_{2}(t)$, $A^{i}_{3}(t)$ and $A^{i}_{4}(t)$, respectively. Therefore, owing to the locally boundedness that we proved in Subsection \ref{SECTURNON} and \ref{SECTURNOFF}, it can be obtained that:
\begin{equation}\label{maxV}
	\begin{aligned}
		\mathop{V}_{t\in\left[t^{i}_{\text{trig}},t^{i+1}_{\text{trig}}\right]}< \max(A^{i}_{1}(t),A^{i}_{2}(t),A^{i}_{3}(t),A^{i}_{4}(t))
	\end{aligned}
\end{equation}
Define the right hand side of equation (\ref{maxV}) as $A^{i}_{\max}(t)$, which can be regarded as a locally maximum upper boundary. By utilizing the relationship in property given in (\ref{AA}), arbitrary local maximum upper boundary $A^{i}_{\max}$ defined on $t\in\left[t^{i}_{\text{trig}},t^{i+1}_{\text{trig}}\right]$ can be extend to a complete one that defined on $t\in\left[0,+\infty\right]$, denoted as $\bar{A}^{i}_{\max}(t)$. This is further illustrated as Figure \ref{EXTEND}.
\begin{figure}[hbt!]
	\centering 
	\includegraphics[width=0.3\textwidth]{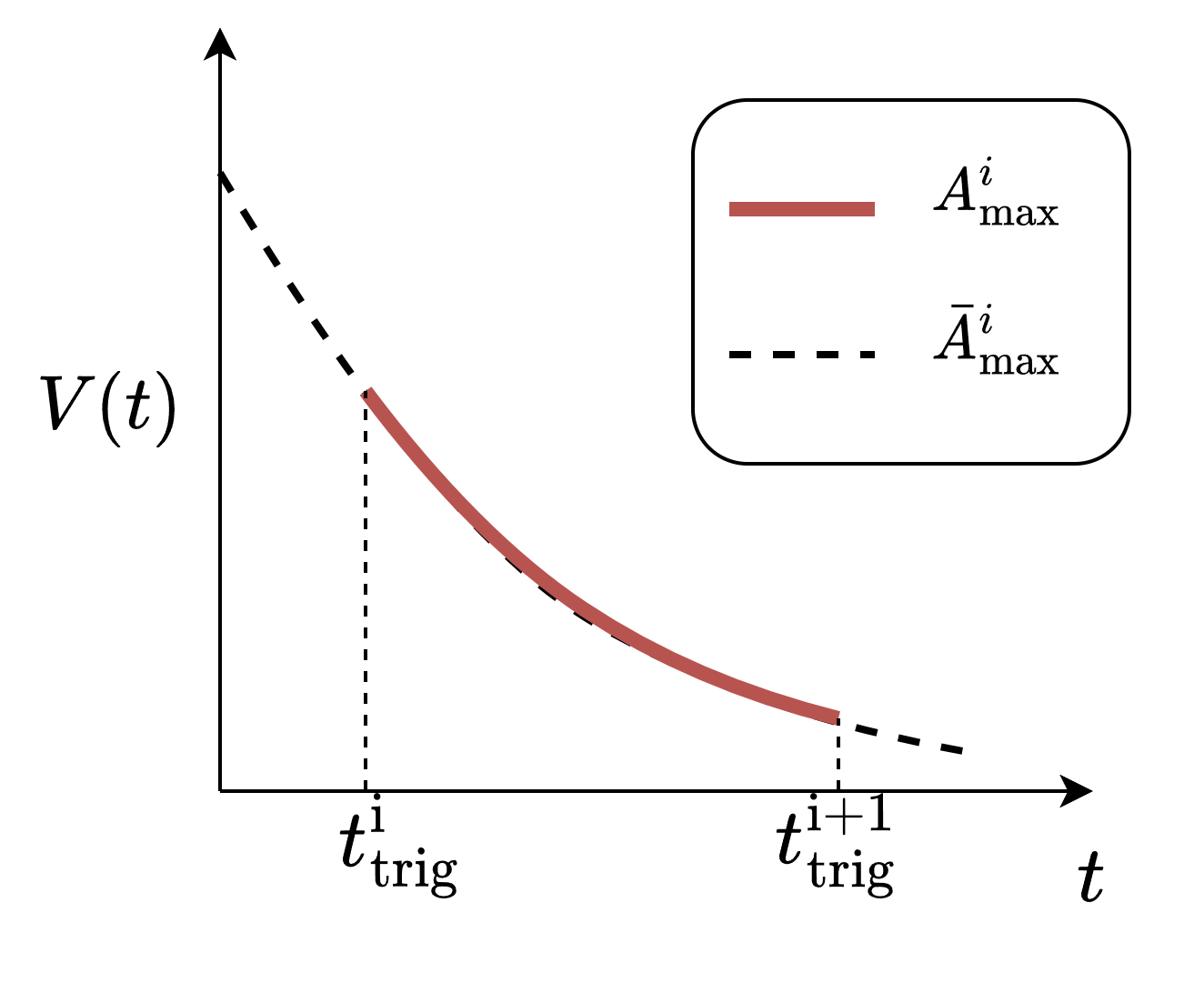}
	\caption{Sketch map of the extending from $A^{i}_{\max}$ to $\bar{A}^{i}_{\max}$}       
	\label{EXTEND}   
\end{figure}

Considering all these trigger events and define a function as:
\begin{equation}
	\mathcal{U}_{V}(t) \triangleq \max(\bar{A}^{i}_{\max}(t))(i=1,2,3,...)
\end{equation}
Obviously, we have:
\begin{equation}
	\mathop{V}_{t\in\left[0,+\infty\right)}< \mathcal{U}_{V}(t)
\end{equation}
Note that $\mathcal{U}_{V}(t)$ is a continuous function that consisted of many segmented exponentially-converged functions, as further explained by Figure \ref{UU}. It can be observed that $\mathcal{U}_{V}(t)$ is consisted of many parts of $\bar{A}^{i}_{\max}$ on different time intervals.
Since $\bar{A}^{i}_{\max}(i=1,2,3,4....)$ are all smooth continuous strictly decreasing exponentially-converged functions, thus $\mathcal{U}_{V}(t)$ is a uniformed upper boundary of $V(t)$.
Meanwhile, $\mathcal{U}_{V}(t)$ converges with an explicitly expressed residual set, such that:
\begin{equation}\label{RESFINAL}
	\lim_{i\to+\infty}\mathcal{U}_{V}(t) = \max(\frac{N_{k}}{C_{t}},\frac{N_{f}}{C_{d}}+\frac{\lambda_{J\max}\rho_{\infty}}{2})
\end{equation}
\begin{figure}[hbt!]
	\centering 
	\includegraphics[width=0.35\textwidth]{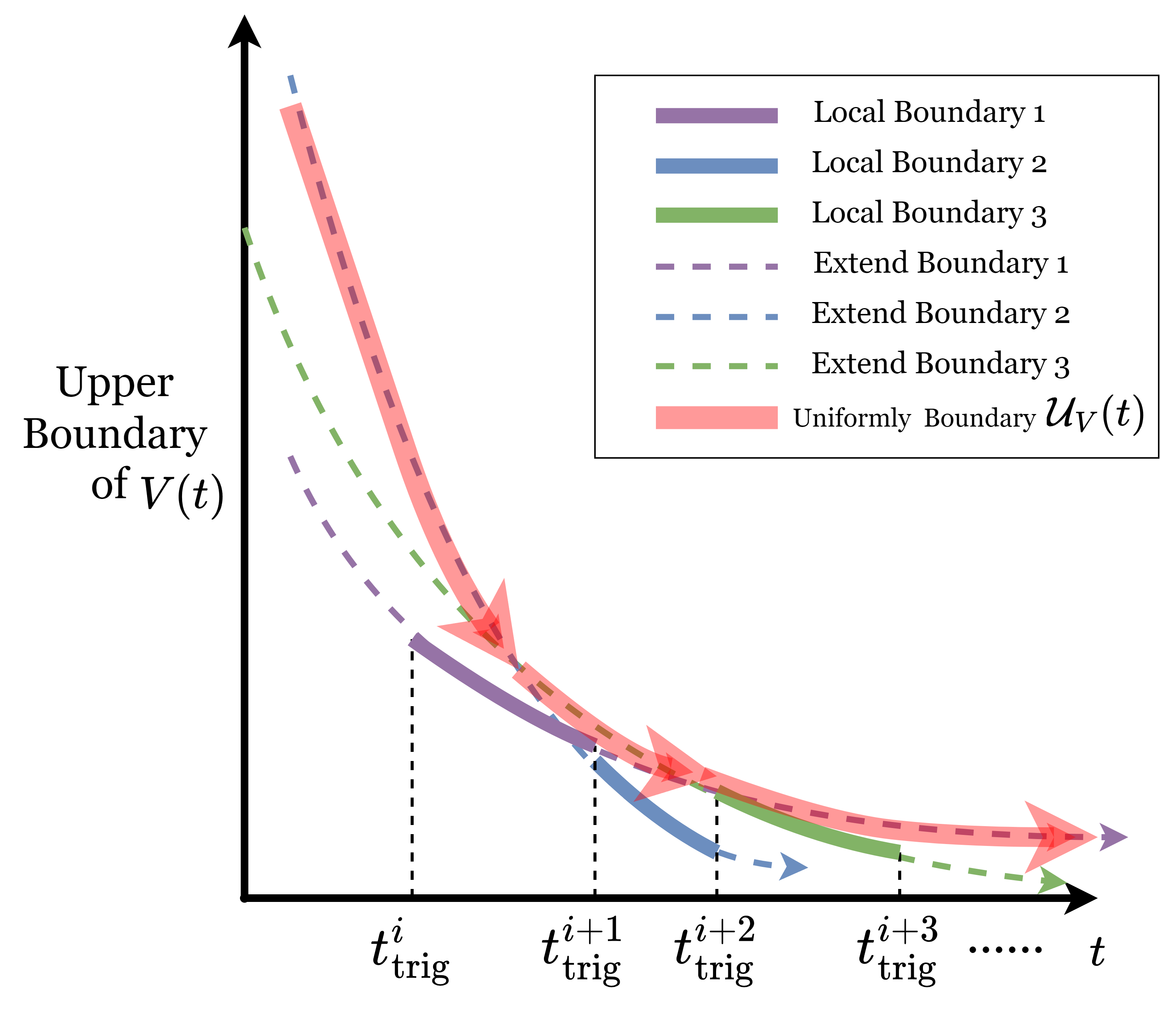}
	\caption{Sketch map of the uniformly upper boundary of $V(t)$, i.e., the $\mathcal{U}_{V}(t)$. These local bounds are separately placed, However, they all can be extend to a complete one (refers to the dotted line), which further constructs a continuous upper bound $\mathcal{U}_{V}(t)$}       
	\label{UU}   
\end{figure}
this suggests that the system is uniformly bounded.

\subsection{Suggestion for Parameter Selecting}\label{PARA}

According to the stated Stability Proof, we further discuss on main factors that govern system's behavior. 

Firstly, we consider the value of $C_{t}$ that defined as $C_{t} = \min\left(\min(|q_{e0}|B_{\omega},\frac{2C_{1}}{\lambda_{J\max}}),2C_{2}\right)$. Since $C_{1}, C_{2} > 0$ are all design parameters that can be defined to be sufficiently large to facilitate the stability, hence $C_{t}$ will be mainly determined by $\min(|q_{e0}|)B_{\omega}$, which is restricted by the given virtual control law's upper boundary $M_{\omega}$ and $|q_{e0}(t)|$'s initial condition. Therefore, practically speaking, $C_{t} < \beta$, $C_{d} < \gamma$ will be the most possible circumstance. Thus, a relatively big $\beta$ and $\gamma$ can be chosen accordingly to make $M(t^{\text{on}}_{k})$ and $M_{f}(t^{\text{off}}_{k})$ decreasing rapidly, thereby achieving a faster convergence behavior of the system. Notably, the value of $\beta$ and $\gamma$ should not be too large, or it may result in a too frequently trigger event. Accordingly, a trade-off should be made to balance the convergence rate and the trigger frequency.

Meanwhile, considering the residual set of $V$ that depicted in equation (\ref{RESFINAL}), it is mainly governed by $N_{k}$, $N_{f}$, and $\rho_{\infty}$. Since the system shows less accuracy when actuators are turned off, the maxima of the upper boundary will be mainly determined by $\frac{N_{f}}{C_{d}}+\frac{\lambda_{J\max}\rho_{\infty}}{2}$. Accordingly, a smaller $\rho_{\infty}$ can be chosen to obtain a low-frequency trigger of turn-on behavior. However, note that a small $\rho_{\infty}$ will improve the convergence accuracy of the system. This indicates that a trade-off should be made to make a balance between the trigger counts and the control accuracy. 

\bibliographystyle{IEEEtran}
\bibliography{etc_bib.bib}

\begin{IEEEbiography}
	[{\includegraphics[width=1in,height=1.25in,clip,keepaspectratio]{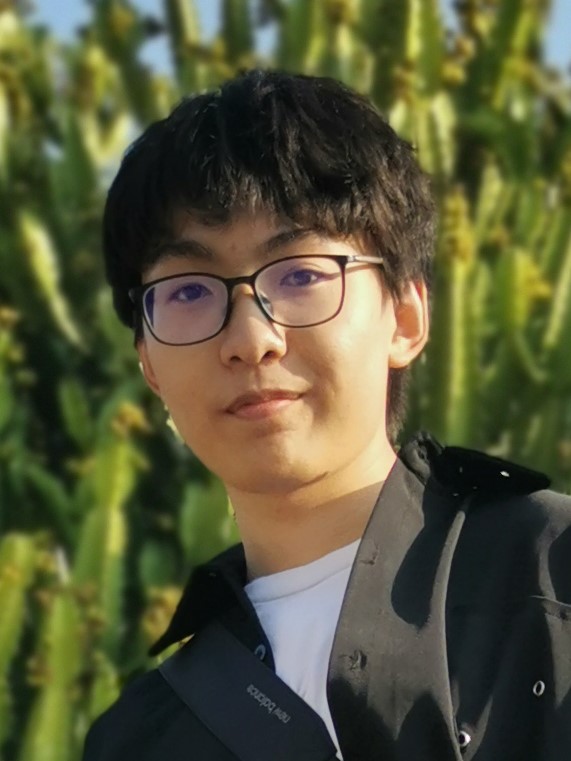}}]
	{Jiakun Lei}{\space} 
	received the B.S. degree in Automatic Control, from the University of Electronic Science and Technology of China(UESTC), Chengdu, China, in 2019. He is working toward a Ph.D. in aeronautical and astronautical science and technology at Zhejiang University, Hangzhou, China.
	His research interests include constrained attitude control, nonlinear hybrid control methodology, and attitude control of spacecraft with complex structures.
\end{IEEEbiography}

\begin{IEEEbiography}
	[{\includegraphics[width=1in,height=1.25in,clip,keepaspectratio]{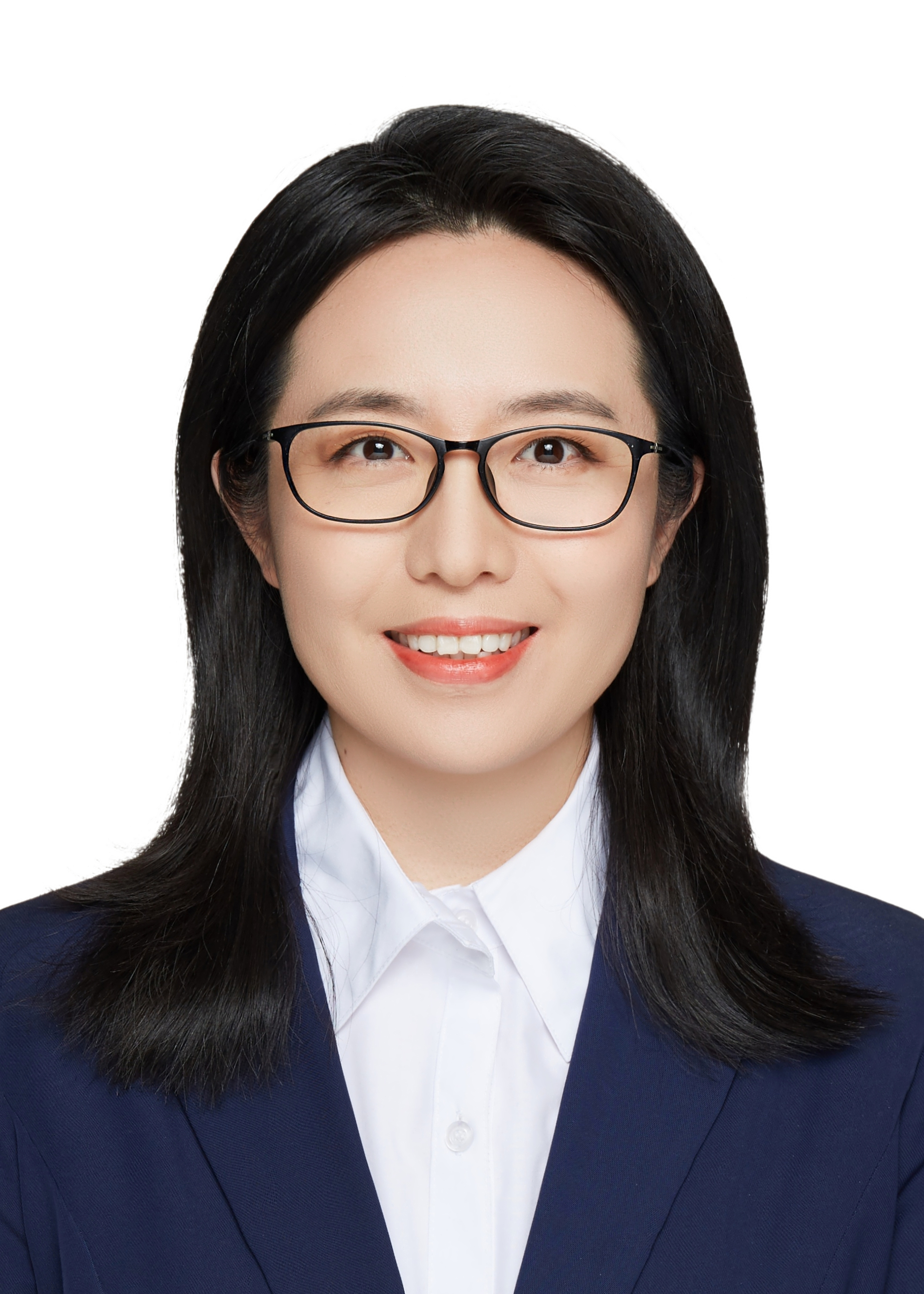}}]
	{Tao Meng}{\space}
	received the B.S. degree in Electronic science and technology, Zhejiang University, Hangzhou, China, in 2004, the M.S. degree in Electronic science and technology, Zhejiang University, Hangzhou, China, in 2006, and the Ph.D. degree in Electronic science and technology, Zhejiang University, Hangzhou, China, in 2009. She is currently a Professor at the School of Aeronautics and Astronautics. Her research interest includes attitude control, orbital control, and constellation formation control of micro-satellite.
\end{IEEEbiography}

\begin{IEEEbiography}
	[{\includegraphics[width=1in,height=1.25in,clip,keepaspectratio]{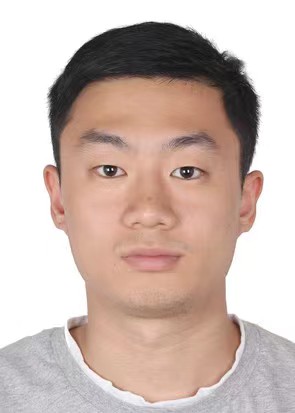}}]
	{Kun Wang}{\space}
	received the B.S. degree from college of control science and engineering, Zhejiang University, Hangzhou, China, in 2019. He is currently working toward the Ph.D. degree in aeronautical and astronautical science and technology in Zhejiang University, Hangzhou, China. His research interests include spacecraft 6-DOF control, spacecraft safety critical control and spacecraft formation control. 
\end{IEEEbiography}

\begin{IEEEbiography}
	[{\includegraphics[width=1in,height=1.25in,clip,keepaspectratio]{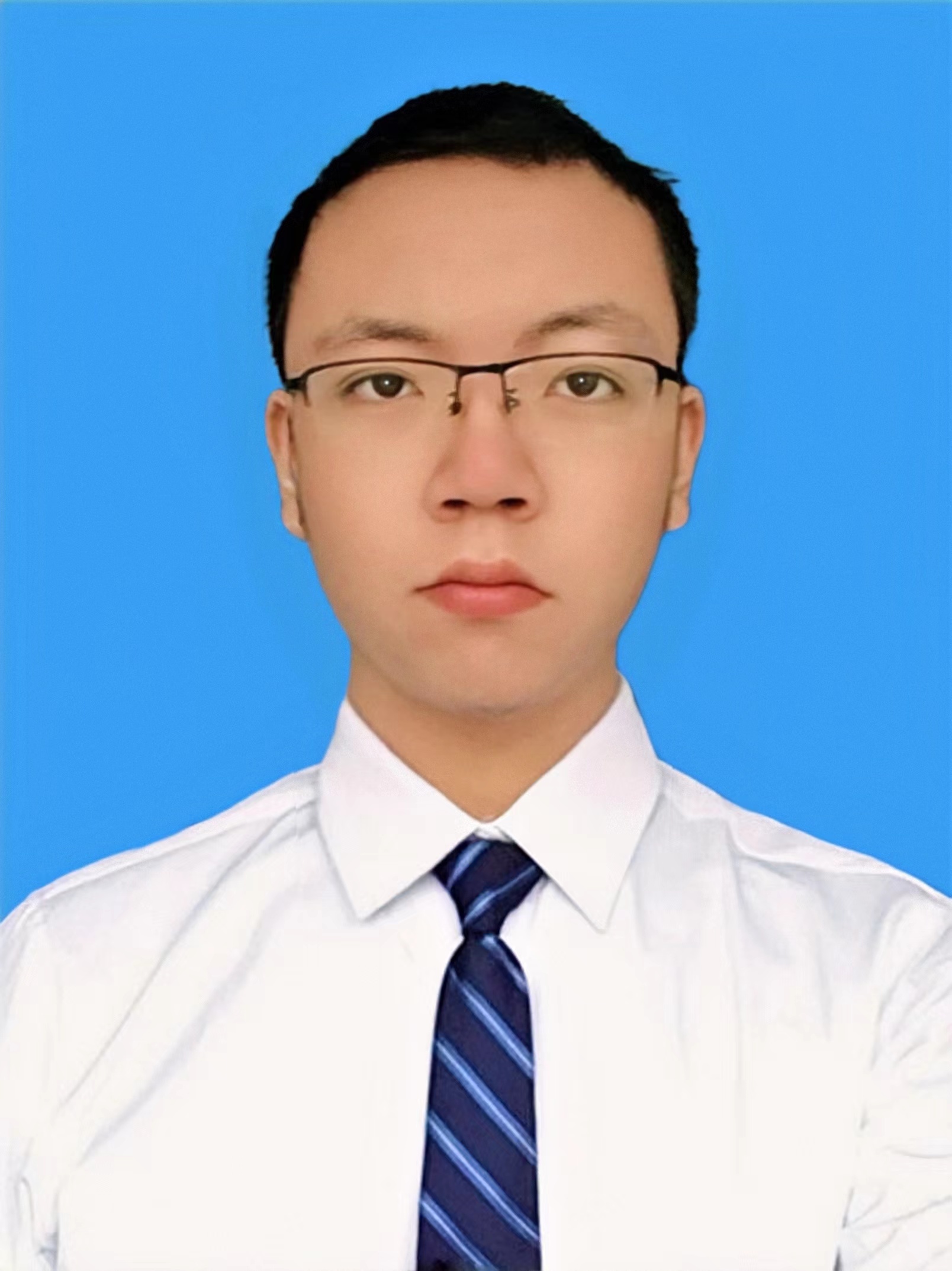}}]
	{Weijia Wang}{\space}
	received the B.S. degree in Aerospace Engineering, from the University of Electronic Science and Technology of China in 2020. He is working toward a Ph.D. in aeronautical and astronautical science and technology at Zhejiang University, Hangzhou, China. His research interests include model predictive control and learning-based adaptive control for 6-DOF spacecraft formation.
\end{IEEEbiography}

\begin{IEEEbiography}
	[{\includegraphics[width=1in,height=1.25in,clip,keepaspectratio]{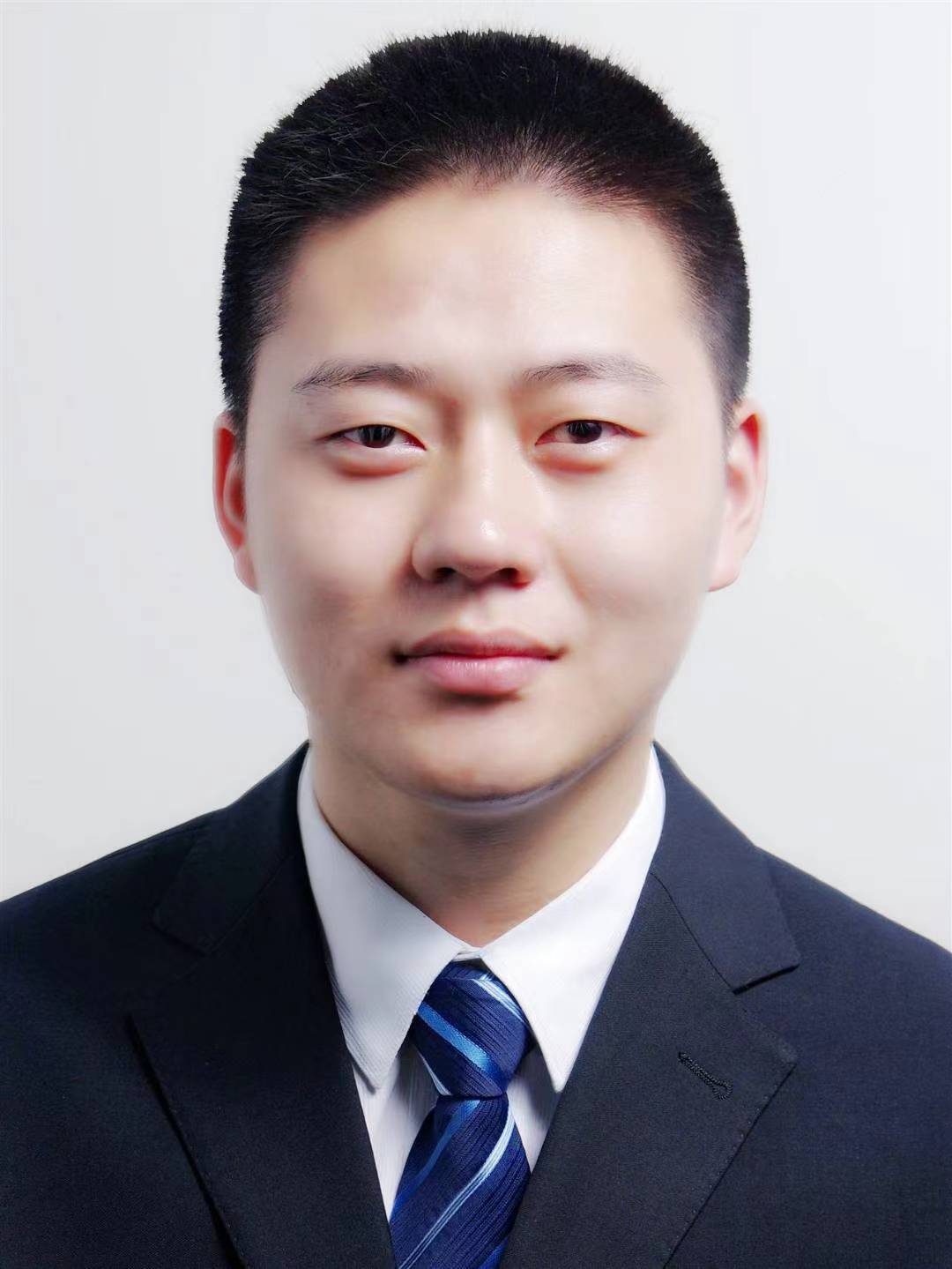}}]
	{Shujian Sun}{\space}
	received the B.S. degree from Electronic Information Engineering (Underwater Acoustic), Harbin Engineering University, Harbin, China, in 2013, and the Ph.D. degree from Electronic Science and Technology, Zhejiang University, Hangzhou, China, in 2020. He is recently the Assistant Professor of School of Aeronautics and Astronautics of Zhejiang University. His research interest include orbit control and formation flying of micro-satellite and micro-propulsion technology.
\end{IEEEbiography}

\end{document}